\documentclass{article}
\usepackage{amssymb}
\usepackage{graphicx}
\usepackage{amsfonts,amsmath}

\usepackage{amsthm}
\newtheorem{lemma}{Lemma}

\usepackage{color}

\usepackage[lined,algoruled,noline]{algorithm2e}

\newcommand{\dd}{\mathrm{d}}
\newcommand{\tsfrac}[2]{{\textstyle \frac{#1}{#2}}}
\newcommand{\nrm}  [1] {\Vert #1 \Vert}

\newcommand{\trans}{^\text{T}}

\newcommand{\Normal}{{\operatorname{N}}}

\newcommand{\E}{{\operatorname{E}}}

\begin{document}

\title{Fast sampling in  a linear-Gaussian inverse problem}

\author{Colin Fox and Richard A. Norton}

\maketitle
\begin{abstract}
We solve the inverse problem of deblurring a pixelized image of Jupiter using regularized deconvolution and by sample-based Bayesian inference. By efficiently  sampling the marginal posterior distribution for hyperparameters, then the full conditional for the deblurred image, we find that we can evaluate the posterior mean \emph{faster} than regularized inversion, when selection of the regularizing parameter is considered. To our knowledge, this is the first demonstration of sampling and inference that takes less compute time than regularized inversion in an inverse problems. Comparison to random-walk Metropolis-Hastings and block Gibbs MCMC shows that marginal then conditional sampling also outperforms these more common sampling algorithms, having better scaling with problem size. When problem-specific computations are feasible the asymptotic cost of an independent sample is one linear solve, implying that sample-based Bayesian inference may be performed directly over function spaces, when that limit exists. 
\end{abstract}


\section{Introduction}
We consider solving a problem in image deblurring using the two frameworks of regularized inversion and sample-based Bayesian inference. 
The computational cost of a standard efficient implementation of regularized inversion is taken as a benchmark against which we compare the cost of drawing samples from the associated Bayesian posterior distribution using algorithms designed for hierarchical Bayesian models, including the \emph{marginal then conditional} sampler. A detailed comparison of algorithms is presented in an example where simplifying assumptions allow all matrices to be diagonalized using the Fourier transform. We also compute the posterior mean image in a second example with improved prior modeling and without simplifying assumptions to demonstrate that the marginal then conditional sampler is feasible, and outperforms regularized  inversion, in the  general setting.

We choose image deblurring because it is a canonical linear inverse problem, when using a low-level representation for the unknown image\footnote{A classification of Bayesian image representations and prior models as low-level, intermediate-level, and high-level is given in~\cite{HurnHusbyRueHSSS}.}. In the idealized setting where data and `true' image are functions, and under the usual assumption that the blurring process is linear and shift-invariant, the forward map corresponds to convolution of the true image with a fixed point-spread function, hence is a Fredholm integral of the first kind. When the point-spread function is square integrable the forward map is Hilbert-Schmidt, hence compact, implying that the inverse problem is ill-posed~\cite{Young,Vogel}.  The discrete problem also displays these properties, having an ill-conditioned forward map, and is referred to as a \emph{discrete ill-posed problem}~\cite{Hansen:1999}.

The use of a low-level representation, in this case a gray-scale pixel image, allows image space to be identified with $\mathbb{R}^n$ and given a normed linear structure so the inverse problem is amenable to regularized inversion. Low-level representations also occur in \emph{exploratory} Bayesian analyses, though more developed applications typically gain significantly through problem specific intermediate- or high-level models (see e.g. \cite{WatzenigFox, HurnHusbyRueHSSS, GrenanderMiller2007}). Low-level image representations effectively restrict the prior information that can be imposed to correlations between pixels values within neighborhoods in the image. We follow Vogel~\cite{Vogel} by using the 2-norm of the graph Laplacian as the regularizing functional, to impose smoothness, and follow Bardsley~\cite{BardsleyRTO} by also using this semi-norm as the negative log prior distribution (up to an additive constant). 
In our experience, regularized inversion and the posterior mean over a low-level image model produce solutions of similar quality. Accordingly, we are primarily interested in computational cost, and only to a lesser degree with quality of reconstructions or measures of uncertainty.

Regularization has the traditional advantage of being implemented using mature and computationally efficient steps, particularly for deblurring by Fourier deconvolution as presented in our first computed example.  On the other hand, sample-based Bayesian inference has the advantage of producing unbiased estimates of the unknown `true' image, or properties of that image, even when evaluated using a \emph{single} posterior sample, while multiple samples allow evaluation of valid uncertainties. Further, Bayesian inference naturally includes estimation of, or averaging over, the effective regularization parameter whereas regularization methods require an extra procedure for determining the regularizing parameter~\cite{Hansen:1999,Vogel}.

The Bayesian formulation is naturally stated as a hierarchical stochastic model that relates measured data, unknown image, and hyperparameters. Typical sample-based methods for exploring the Bayesian posterior distribution utilize random-walk MCMC over the posterior distribution, which can be very slow, especially when compared to efficient computation of a regularized solution.

However, as we show here, the stochastic model may be factorized by marginalizing over the unknown image, and we need only run a general MCMC for the low-dimensional marginal distribution over hyperparameters; our main contribution is to observe that the ratio of high-dimensional determinants required for that low-dimensional MCMC can be computed cheaply, which we demonstrate theoretically and in the computed examples. Drawing an \emph{independent} sample from the posterior distribution is then dominated by the same linear solve required in regularized deconvolution. Since selection of the regularizing parameter requires many such linear solves, the computational cost of the regularized solution is equivalent to drawing many independent posterior samples, from which robust estimates may be evaluated along with quantified uncertainties. Thus, in a setting where regularized inversion is often applied as an efficient solution method, we demonstrate that Bayesian inference over an equivalent model is actually cheaper. The marginal then conditional sampler is also cheaper than the block Gibbs sampler that has recently been presented as an efficient sampler for the linear-Gaussian inverse problem.  Further, because the MCMC we implement over hyperparameters can be made independent of image size, that MCMC can be performed directly on the infinite dimensional image model, when that limit is well defined, with image-size dependence only occurring in the setup phase and in the final image-forming step.

The paper is structured as follows: In Section~\ref{sec:deconv} we present a problem in semi-blind deconvolution, and in Sections~\ref{sec:reg} and \ref{sec:Bayes} we present formulations for its solution via regularized inversion and Bayesian inference, respectively. Subsection~\ref{sec:mtc} introduces the marginal then conditional sampler that allows efficient computation. Section~\ref{sec:comparative} compares sampling algorithms for  linear-Gaussian inverse problems, including the block Gibbs sampler,  the one-block algorithm, and the marginal then conditional sampler. Numerical comparison of all algorithms is presented in Section~\ref{sec:periodic} to validate theoretical results, in the convenient setting where efficient computation is available via the FFT. Numerical implementation of fast sampling for a more general image model, including nuisance pixels and not assuming periodic boundary conditions, is presented in Section~\ref{sec:nonperiodic}.  Section~\ref{sec:discussion} presents a discussion of results and implications for the infinite dimensional limit. Technical calculations that we use for efficient marginal sampling in Section~\ref{sec:periodic} are presented in the Appendix.

\subsection{An example of semi-blind deconvolution}
\label{sec:deconv}

Figure \ref{fig:blurry-jupiter} contains a photograph of Jupiter
taken in the methane band (780nm) on a grid of size $256\times256$
pixels, each takes an integer value from 0 to 255. As can be seen,
the image is somewhat blurry; the inverse problem is to recover the `true' unblurred version of the image. 

\begin{figure}
\centerline{\includegraphics[scale=0.5]{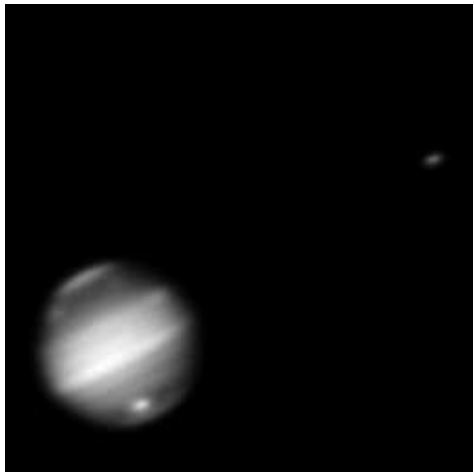}}
\caption{A blurry photograph of Jupiter taken in the methane band (780nm).\label{fig:blurry-jupiter}}
\end{figure}

We model the blurring process by convolution with a fixed point-spread function that we denote $k$. Denoting the 
true unknown image by $x$ and the data by $y$, the observation process may be written
\begin{eqnarray}
y & = & k\ast x+\eta\nonumber \\
 & = & Ax+\eta \label{eq:A}
\end{eqnarray}
where the forward map $A$ is the linear operator representing convolution and $\eta$ is  an unknown `noise' vector representing measurement errors including digitization. In the convolution form we think of $x$ and $y$ as images in the plane, while in the operator form it is more convenient to write $x$ and $y$ as vectors, with observations $y\in\mathbb{R}^m$ and the unknown $x\in\mathbb{R}^n$. Through most of this paper we compute with images $x$ and $y$ that are $p\times p$ pixels in size, i.e., $n=m = p^2$. We are interested in the practical case where one wishes to perform inference from an existing data set so $m$ is necessarily finite\footnote{In industrial applications, time and money  typically increase with $m$, so smaller is better.}. However, the size of the reconstructed image $n$ is always a modeling choice; we are also interested in how that choice affects computational cost, particularly in the infinite dimensional limit $n\rightarrow\infty$. 

Some authors take \emph{the} infinite dimensional limit to mean that both $n\rightarrow\infty$ and $m\rightarrow\infty$, that is the size of data also tends to infinity, as is common in studies of idealized inverse problems (see, e.g. the classic work~\cite{ColtonKress98}). 
One needs to be aware of the potential confusion, and that results that rest on the assumption $m\rightarrow\infty$, such as those presented in~\cite{NotBardsleyStuart}, may not hold for the case considered here. 

Since the point-spread function $k$ is unknown, and the forward map is convolution, this deblurring problem is often called \emph{blind deconvolution}. However, for the sake of this example we will assume a form for $k$ as follows. The upper right-hand portion of the photograph shows one of the Galilean satellites that is small enough to be considered close to a point source, and so we can use that region of the photograph ($32 \times 32$ pixels), normalized, as an approximation to the point-spread function. Hence, we actually implement  \emph{semi-blind} deconvolution. It is possible to model and infer the point-spread function within regularized inversion or the Bayesian calculation, to implement true blind deconvolution\footnote{We set modeling of the point-spread function and blind deconvolution as a challenge question for students when we use this as a classroom example.}, though we do not consider it further here. 

The forward map $A$ is often computed efficiently using the fast Fourier transform (FFT).  We follow~\cite{BardsleyRTO} by assuming periodic boundary conditions for $x$ in the first computed example (but not in the second). 
The assumption of periodic boundary conditions means that convolution is equivalent to circular convolution, i.e.
\[ Ax \equiv k \circledast x, \]
so $A$ is diagonalized by the discrete Fourier transform and the action of $A$ can be computed in the transform domain for a cost of $n$ multiplications, with a further $\mathcal{O}(n\log n)$ operations for the inverse FFT to the image domain.

The point spread function is non-zero for only $32 \times 32$ pixels and there appears to be a black band (of width $>16$ pixels) around the edge of the blurry photograph. Thus, it is likely that the true image $x$ is black in this $16$ pixel band and hence no significant artifacts will be generated by assuming periodic boundary conditions.  
For images that are not dark at the edges it is usual to zero pad both the data $y$ and recovered image $x$ before calculating the FFT~\cite{NumRec}.

Neither of these applications of the FFT (zero padding or not) gives a complete model of the forward map as neither includes the influence of bright pixels \emph{outside} the image region. It is natural in the Bayesian formulation to include and marginalize over these `nuisance' pixels, which we do in the second computed example.

\section{Solution by Regularization}
\label{sec:reg}

We regularize the inverse of $A$ by penalizing images that are not smooth in the sense that a pixel differs from the average of its neighbors, as in~\cite{Vogel}. 
Define the neighborhood structure using the usual square pixel lattice; the neighbors of pixel $i$ are the pixel locations that are above, below, to the left, and to the right of pixel location $i$. 

Write $j\sim i$ when pixel $j$ is a neighbor of pixel $i$, and $\partial_{i}$ for the set of neighbors of pixel $i$, i.e.
\[
\partial_{i}=\left\{ j\neq i|j\sim i\right\}.
\]
Denote by $\left|\partial_{i}\right|$ the number of neighbors of pixel $i$; $\left|\partial_{i}\right|=4$ for all pixels when periodic boundary conditions are assumed.
Define matrix $L$ to be 
\begin{equation}
\label{eq:L}
L_{ij}=\begin{cases}
\left|\partial_{i}\right| & i=j\\
-1 & j\in\partial_{i}\\
0 & \mbox{otherwise}
\end{cases}
\end{equation}
which is the graph Laplacian on the neighborhood graph.  The action of $L$ is equivalent to convolution with the $5$-point finite-difference stencil~\cite{dahlquistbjork}, i.e.,
\[ Lx \equiv x\ast \left(\begin{array}[]{ccc}
                               & -1 &   \\
                            -1 & 4 & -1 \\
                               & -1 &
                          \end{array} \right).\]
In the periodic case $L$ is also diagonalized by the discrete Fourier transform with efficient calculation possible using the FFT. Even without that assumption, sparsity of $L$ allows efficient operation, also requiring only $\mathcal{O}(n)$ operations.


The regularized estimate for the deblurred image is defined by
\begin{equation}
 	\hat{x}_\lambda = \arg \min_{x} \lVert Ax - y \rVert^2 + \lambda x^T L x
 	\label{eq:varreg}
\end{equation}
for regularization parameter $\lambda\geq 0$. The minimizer $\hat{x}_\lambda$ may be calculated by solving the generalized deconvolution equations 
\begin{equation}
  (A^T A + \lambda L) \hat{x}_\lambda = A^T y. \label{eq:gendeconv}
\end{equation}
With periodic boundary conditions, this linear solve requires $\mathcal{O}(n)$ operations in the transform domain, with one-off $\mathcal{O}(n \log n)$ computing costs when transforming between image and transform domains using the FFT. Without this simplifying assumption, efficient solution of the system in~\eqref{eq:gendeconv} may be performed using a linear solver that exploits the sparsity of $(A^T A + \lambda L)$.

As noted in~\cite{BardsleyRTO}, $-h^{-2}L$ is the discrete Laplacian, where $h$ is the pixel spacing. Since $h\propto 1/p$ then $nL\rightarrow-c\nabla^2$ (with appropriate boundary conditions) for some constant $c$, as $n\rightarrow \infty$. 
The transformations $L\leftarrow nL$ and $\lambda\leftarrow\lambda/n$ leave all terms in the right-hand side of~\eqref{eq:varreg} unaltered, so, for finite $n$, the use of $L$ or the negative discrete Laplacian defines \emph{identical} sets of deblurred images, though with altered value of the regularization parameter. We have used the graph Laplacian, as in other computational work~\cite{Vogel}, to minimize roundoff errors.

We used the L-curve criterion~\cite{HansenOLeary} to select $\lambda$, using Hansen's {\tt l\_corner.m} algorithm in {\tt regtools}~\cite{HansenRegTools} that takes $200$ solves of~\eqref{eq:gendeconv} to find the `corner' of the L-curve. By Parseval's theorem the data misfit and regularization semi-norm may be computed in the transform domain, so each solve requires $\mathcal{O}(n)$ operations.  For our Jupiter example we found a regularization parameter $\lambda = 5.6724 \times 10^{-3}$ in $0.517$ seconds (including $0.507$ seconds for $200$ solves) in MATLAB R2012b using a Lenovo X230 laptop with an Intel CORE i5 processor.  The computed deblurred image and L-curve are in Figure \ref{fig:deconv-jupiter}.  In the L-curve figure, the cluster of crosses correspond to sampled images (more on this later).  

\begin{figure}
\includegraphics[width=0.45\textwidth]{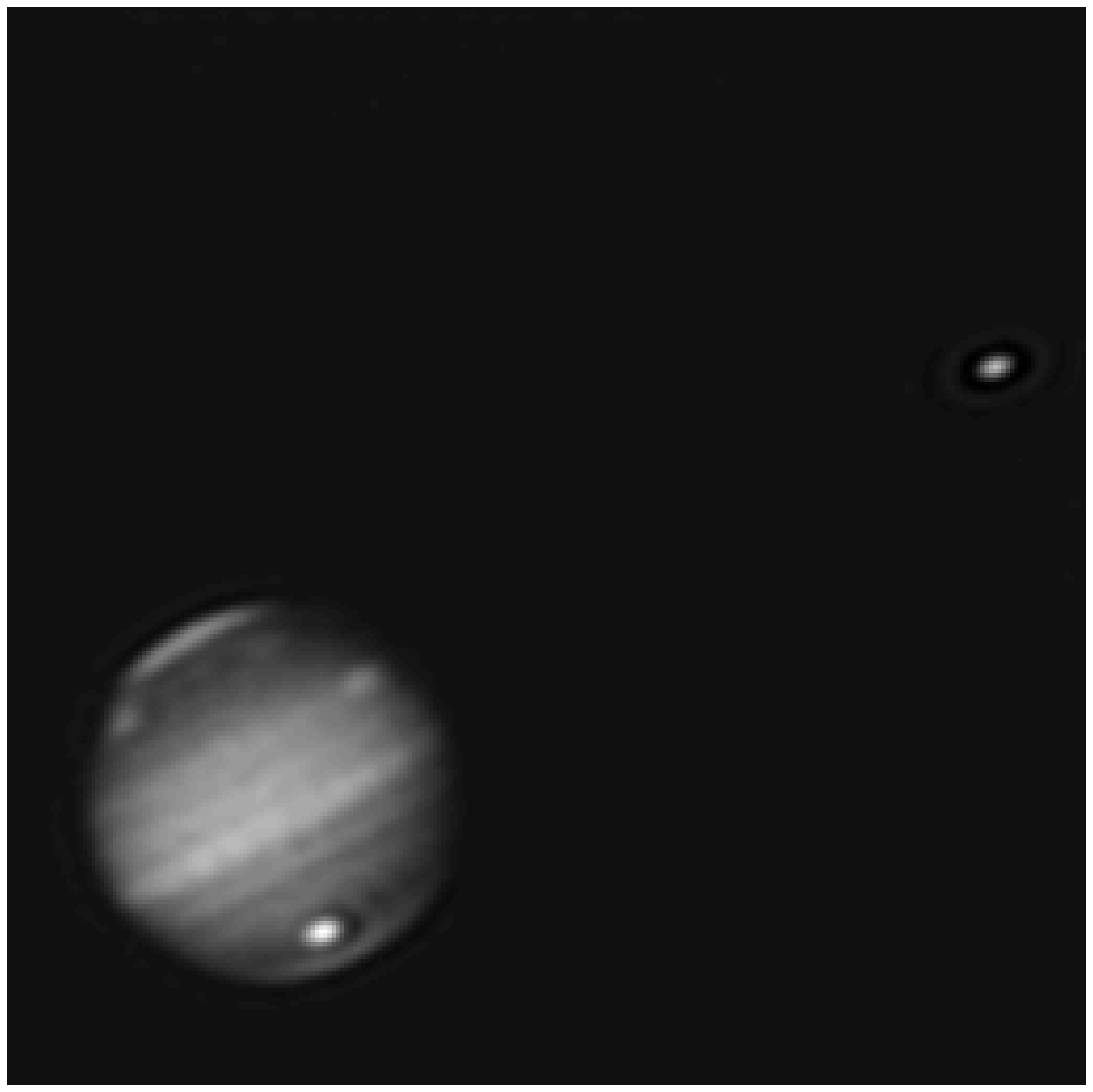}
\includegraphics[width=0.54\textwidth]{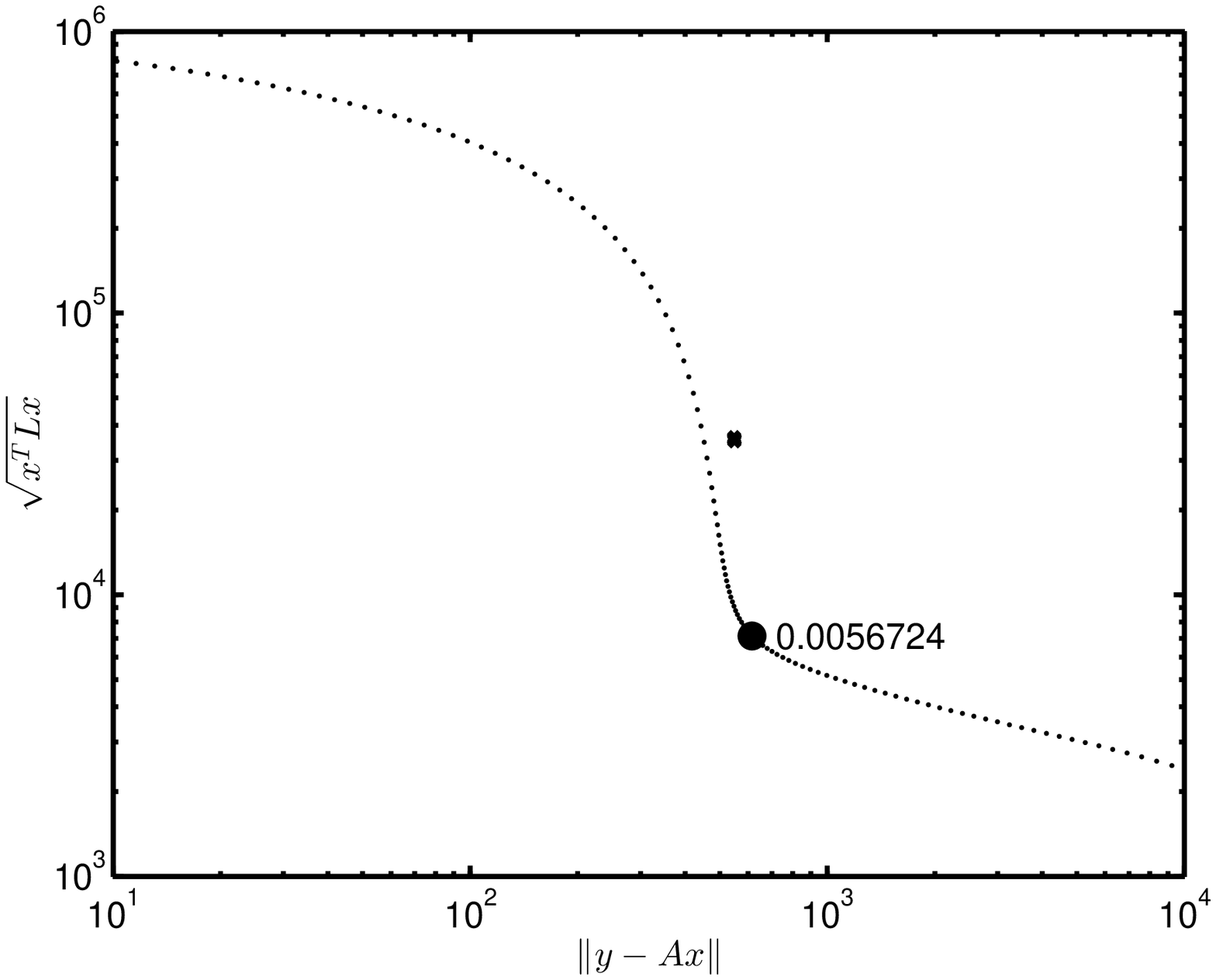}
\caption{Left; the regularization estimate of a deconvolved Jupiter computed with $\lambda = 5.6724 \times 10^{-3}$.  Right; the corresponding L-curve.  The tight cluster of crosses correspond to sampled images.}
\label{fig:deconv-jupiter}
\end{figure}

In the deblurred image of Jupiter in Figure \ref{fig:deconv-jupiter} we see artificial `ringing' phenomena around the satellite and Jupiter.  The ringing effect is reduced with better modeling of the point-spread function, though damping of the solution and some ringing is a typical feature of regularized solutions.

\section{Bayesian modeling and inference}
\label{sec:Bayes}
Bayesian modeling naturally proceeds in a hierarchical manner and allows, one could say  \emph{requires}, the specification of distributions over unknown quantities. The Bayesian inferential framework accommodates very general models for the unknown latent field, allowing representations in arbitrary measure spaces, often drawing on ideas from spatial statistics~\cite{HurnHusbyRueHSSS,BCG} or pattern theory~\cite{GrenanderMiller2007}. Here we restrict representations to the low-level image model used in the regularized solution and reuse model components and make distributional choices to enable a comparison of methods that are closest equivalents. 

\subsection{Hierarchical stochastic model}
Observed data $y$ is related to the unknown true image $x$ via the observation model in~\eqref{eq:A}.  We assume that the components of the  noise vector $\eta$ in~\eqref{eq:A} are independent and identically distributed as a zero mean Gaussian with some variance $\gamma^{-1}$, i.e., $\eta\sim \mathrm{N}\left(0,\gamma^{-1} I\right)$ where $I$ is the identity matrix.  Such a model is common in settings where knowledge of instrumentation justifies the assumption that errors on individual measurements are independent, with zero mean, but with unknown average amplitude\footnote{Studies in inverse problems that assume $m\rightarrow\infty$ have an infinite set of samples of the \emph{noise}, because the forward map is effectively finite rank. Hence, sample estimates of almost any statistic of the noise, including $\gamma$, have zero variance; this is one origin of the reducibility of the Gibbs sampler noted in~\cite{NotBardsleyStuart} that does not occur in the practical case considered here.}. Then the distribution over $y$ conditioned on image $x$ and precision $\gamma$ is
\begin{equation}
y|x,\gamma  \sim  \mbox{N}\left(Ax,\gamma^{-1} I\right).\label{eq:likelihood}
\end{equation}

We are able to prefer smoothness in the unknown `true' image $x$ by modeling $x$ as a draw from a Gaussian Markov random field (GMRF) that 
assigns low probability to non-smooth images. We use a \emph{locally linear} GMRF~\cite{higdon2006} defined by the conditional distributions
\[
x_{i}|x_{\partial_{i}}\sim\mbox{N}\left(\left|\partial_{i}\right|^{-1} {\textstyle \sum_{j\in\partial_{i}}} x_{j},(\delta\left|\partial_{i}\right|)^{-1}\right)
\]
in which $\delta$ is an unknown lumping constant. This models each pixel value as a Gaussian random variable with mean equal to
the average of neighboring pixel values, and with some variance controlled by a lumping constant $\delta$. Thus, the preferred value for each pixel is the average of its neighbors, while the lumping constant controls how strongly that preference is asserted. The joint distribution over the vector $x$ for given $\delta$ is then the (intrinsic) multivariate Gaussian 
\begin{equation}
  x|\delta \sim \mbox{N}(0,(\delta L)^{-1}) \label{eq:prior}
\end{equation}
where the matrix $L$ is the discrete Laplacian defined in~\eqref{eq:L}. When assuming periodic boundary conditions $L$ has a zero eigenvector  (the constant vector) and so the prior is improper. However, the posterior is normalizable so all marginal and conditional distributions are well defined. It is usual to regularize the inverse in~\eqref{eq:prior} by adding a small constant (`a nugget') to the diagonal of $L$~\cite{higdon2006}. Alternatively, we may define $L^{-1}$ to be the Moore-Penrose inverse of $L$ in which case the constant vector is also a zero eigenvalue of the covariance. 


The two parameters $\gamma$ and $\delta$ are assumed unknown and so we must also specify distributions over these as well. 
For the present we will define the vector parameter $\theta = (\gamma,\delta)$ and simply write the joint distribution as $\pi(\theta)$.

Combining the stochastic models~\eqref{eq:likelihood}, \eqref{eq:prior}, and for $\theta$, we may write the linear-Gaussian Bayesian model in the slightly more general form
\begin{subequations}
\label{eqs:lgbm}
\begin{eqnarray}
y|x,\theta & \sim & \mbox{N}\left(Ax,\Sigma\left(\theta\right)\right)\label{eq:yxtheta}\\
x|\theta & \sim & \mbox{N}\left(\mu,Q^{-1}\left(\theta\right)\right)\label{eq:xtheta}\\
\theta & \sim & \pi\left(\theta\right).\label{eq:theta}
\end{eqnarray}
\end{subequations}
This hierarchical stochastic model occurs commonly in statistics~\cite{SimpsonLindgrenRue2012}, in which $y$ is \emph{observed data}, $x$ is a \emph{latent field} with mean $\mu$, and $\theta$ is a vector of \emph{hyperparameters} that model uncertainties in the measurement noise covariance $\Sigma$ and in modeling of the precision (inverse of covariance) matrix $Q$ of the latent field process. In the language of Bayesian analysis, \eqref{eq:yxtheta} defines the \emph{likelihood function} for unknown $x$ and $\theta$ once data $y$ is observed, \eqref{eq:xtheta} is the \emph{prior distribution} over latent field $x$ with hyperparameters $\theta$, and \eqref{eq:theta} sets the \emph{hyperprior distribution} over those hyperparameters.

The prior distribution used in~\eqref{eq:prior} has played a role in some of the earliest developments in Bayesian statistics. While a comprehensive discussion on (hyper)prior distributions is well outside the scope of this paper, it is interesting to note a few details that are relevant to our study.

The parameter $\delta$ in~\eqref{eq:prior} is a positive scale parameter, whose numerical value depends on the units chosen for the spacing between pixels; a change of units corresponds to the transformations $L\leftarrow cL$ and $\delta\leftarrow\delta/c$, for some $c>0$. Jeffreys~\cite{Jeffreys453} addressed the question of how to set the functional form of $\pi(\delta)$ so that inference over $x$ is independent of the choice of units, and developed what is now called the Jeffereys scale prior\footnote{Jeffreys priors now form a general class of prior distributions (see, e.g.~\cite{RobertChopinRousseau2009}), often called reference priors, with the modern development under the moniker of \emph{objective priors}~\cite{BergerBernardoSun}.} that gives  \emph{identical} marginal posterior distributions for $x|y$, whatever the scaling of $\delta$. Hence, the Jeffreys scale prior is often viewed as being \emph{uninformative} about the units of $\delta$ since inference on $x$ is independent of that choice. In the Jupiter problem, the difference between using the graph Laplacian $L$ or the negative discrete Laplacian $nL$, for finite $n$, is a scaling of $\delta$ by $n$, just as we observed for scaling $\lambda$ in Section~\ref{sec:reg}. Hence, the Jeffreys scale prior produces inference for $x$ that is identical whether $L$ or $nL$ is used to define the prior precision in~\eqref{eq:prior}. 

Bardsley~\cite{BardsleyRTO} employed a conjugate prior distribution over $\delta$ to enable Gibbs sampling, and observed (\cite[\S 4.4]{BardsleyRTO}) that posterior inference differed when using $L$ versus $nL$. From the perspective of Jeffreys priors this difference is a consequence of the conjugate prior being \emph{informative} with respect to the scaling of $\delta$, including units. 
The same choice and problem occurs in~\cite{NotBardsleyStuart}. 

Simply using the Jeffreys prior in high dimensional settings, such as inverse problems, can lead to significant unintended biases~\cite{RobertChopinRousseau2009}; see also~\cite{WatzenigFox,NichollsJones2001} for examples and resolution. 
We consider a very positive feature of the MTC sampler developed here is that it can operate with \emph{any} hyperprior distribution $\pi(\theta)$ (that can be evaluated) and so the algorithm does not prejudice those considerations. 

\subsection{Posterior inference}
The focus of inference is the \emph{posterior distribution} over unknowns $x$
and $\theta$ conditioned on measured $y$, given by Bayes' rule as
\begin{equation}
\pi\left(x,\theta|y\right)=\frac{\pi\left(y|x,\theta\right)\pi\left(x,\theta\right)}{\pi\left(y\right)}.
\label{eq:posterior}
\end{equation}
Note that we are performing the standard abuse of notation by using the symbol $\pi$ to denote any probability density function, and associated distributions, with the particular function determined implicitly by the arguments. Solutions and uncertainties may be computed as the posterior expectation of some function $h$ of $x$,
\[
\text{E}_{x,\theta |y}\left[ h\left(x \right) \right] =\int h\left(x
\right) \pi \left(x,\theta |y\right)  \,\dd x\, \dd\theta
\]
which implicitly averages over the nuisance parameter $\theta$. Sample-based methods use a Monte Carlo estimate of the integral. When $(x,\theta)^{(1)}, \ldots,(x,\theta)^{(N)}\sim \pi \left(x,\theta |y\right)$ are iterates of an ergodic Markov chain,
\[\text{E}_{x,\theta |y}\left[ h\left(x \right) \right] \approx
\frac{1}{N}\sum_{i=1}^{N}h\left(x^{(i)} \right) \]
with convergence guaranteed by a central limit theorem~\cite{kipnis1986}.

For computation, it is important to observe that the numerator in~\eqref{eq:posterior}
\[ \pi\left(y|x,\theta\right)\pi\left(x,\theta\right) = \pi\left(y|x,\theta\right)\pi\left(x|\theta\right)\pi\left(\theta\right) \]
may be evaluated as the product of the three density functions in the hierarchical model~\eqref{eqs:lgbm}. However the normalizing constant
\[\pi(y)=\int\int \pi \left(y, x,\theta\right)\,\dd x\,\dd\theta\]
is typically not available in the sense that it is infeasible to compute. We will assume throughout that $\pi(y)$ is finite, and non-zero. Hence the posterior density may be evaluated up to an unknown constant, and therefore can be explored using Metropolis-Hastings MCMC~\cite{GRS_MCMCIntro}. Indeed this is the most common method of sampling from the posterior distribution being easy to implement, typically employing a random-walk proposal distribution. Examples of such methods include the adaptive Metropolis (AM), Metropolis-adjusted Langevin (MALA) and hybrid (or Hamiltonian) Monte Carlo (HMC) methods, amongst many others. Random-walk Metropolis-Hastings MCMC methods on the full state space are typically very slow to converge, not uncommonly requiring $10^4$ to $10^5$ iterations for a single independent sample \cite{Debski, FoxNicholls97}, with each iteration requiring simulation of high-dimensional data over a high-dimensional image space. We do not implement such a calculation here as the marginal then conditional algorithm we present next is several orders of magnitude cheaper.

\subsection{Marginal then conditional sampling}
\label{sec:mtc}
We propose to significantly speed up sampling by first sampling from the marginal posterior distribution over hyperparameters $\theta$
\[ \pi\left(\theta|y\right) = \int \pi\left(x,\theta|y\right)\,\dd x \]
then from the full conditional distribution over $x$, to give Algorithm~\ref{alg:margcond}, that we call marginal then conditional (MTC) sampling.  In statistics, this algorithm is widely known as the marginal algorithm, see e.g. \cite{vanDyk,PR2008}. We prefer the more descriptive term, and use it to refer to both the decomposition and the computational scheme we present later. 
\begin{algorithm}
\label{alg:margcond}

draw $\theta \sim \pi\left(\theta|y\right)$

draw $x  \sim  \pi \left(x|y,\theta\right)$

\caption{MTC sampling from the posterior distribution}
\end{algorithm}

\begin{lemma}
\label{lem:mtc}
Algorithm~\ref{alg:margcond} generates a sample from the posterior distribution, i.e., 
\[ \left(x ,\theta\right) \sim \pi\left(x,\theta|y\right). \]
\end{lemma}
\begin{proof}
The density function over $x$ and $\theta$ is $ \pi\left(x|y,\theta\right)\pi\left(\theta|y\right) = \pi\left(x,\theta|y\right)$.
\end{proof}

When the samples over $\pi\left(\theta|y\right)$ in Algorithm~\ref{alg:margcond} are \emph{independent}, then so are the posterior samples $\left(x ,\theta\right)$. The case where $\theta \sim \pi\left(\theta|y\right)$ is generated by one step of a (geometrically ergodic) MCMC was considered by Acosta, Huber \& Jones~\cite{AcostaHuberJones} who called Algorithm~\ref{alg:margcond} a `linchpin variable sampler' with $\theta$ being the `linchpin variable'. They showed that the convergence rate of the chain in $\left(x ,\theta\right)$ is the same as the chain in $\theta$. Later, we will use a MCMC to draw samples $\theta \sim \pi\left(\theta|y\right)$; however, that iteration is sufficiently fast that we will take many steps of the MCMC to generate an effectively independent $\theta$ before drawing $x  \sim  \pi \left(x|y,\theta\right)$. The observation that independent $\theta$ gives independent $\left(x ,\theta\right)$ is a degenerate form of the result in~\cite{AcostaHuberJones}.

\subsubsection{Marginal posterior for $\theta$}{~}\\

\begin{lemma}
\label{lem:marg}
\[
\pi\left(\theta|y\right)=\frac{\pi\left(y|\theta,x\right)\pi(x|\theta) \pi\left(\theta\right)}{\pi\left(x|\theta,y\right)\pi\left(y\right)}
\]
\end{lemma}
\begin{proof}
$\pi\left(x,y,\theta\right)=\pi\left(x|\theta,y\right) \pi(y|\theta) \pi\left(\theta\right)$ and $\pi(x,y,\theta) = \pi(y|x,\theta)\pi(x|\theta)\pi(\theta)$.  Writing $\pi(y|\theta)\pi(\theta)=\pi(\theta|y)\pi(y)$ and using $\pi(y) \neq 0$, the result follows.
\end{proof}

Since $\pi\left(y\right)$ does not depend on $\theta$, it follows that
\begin{equation}
\pi\left(\theta|y\right)\propto\frac{\pi\left(y|\theta,x\right)\pi(x|\theta)\pi\left(\theta\right)}{\pi\left(x|\theta,y\right)}.\label{eq:margtheta1}
\end{equation}
This result is given for Gaussian distributions in~\cite{SimpsonLindgrenRue2012}, though it holds more generally as~\eqref{eq:margtheta1} shows. In principle the right-hand side of~\eqref{eq:margtheta1} may be evaluated using any value of $x$ for which $\pi\left(x|\theta,y\right)$ is appreciable enough to avoid round-off issues. In the linear-Gaussian case~\eqref{eqs:lgbm} we can eliminate $x$ to give
\begin{flalign}
& \pi\left(\theta|y\right) \propto 
\sqrt{ \frac{\det(\Sigma^{-1})\det(Q)}{\det(Q+A^T\Sigma^{-1}A)}}  &\nonumber\\
& \exp \left\{ -\frac{1}{2} (y-A\mu)^T \left[ \Sigma^{-1}  - \Sigma^{-1} A (A^T\Sigma^{-1}A + Q)^{-1} A^T \Sigma^{-1} \right]  (y-A\mu) \right\} \pi(\theta).\label{eq:overlooked}&
\end{flalign}
This reduction appears to have been overlooked in~\cite{SimpsonLindgrenRue2012}. Note that this distribution is not Gaussian because of the dependence of $\Sigma$ and $Q$ on $\theta$ (that is not shown for brevity), and that arbitrary hyperprior distribution $\pi(\theta)$ is allowed.

For our Jupiter example this formula simplifies to
\begin{equation}
  \pi(\theta|y) \propto \delta^{n/2}
	\exp \left( -\frac{1}{2} g\left(\lambda \right) - \frac{\gamma}{2} f\left(\lambda \right) \right) \pi(\theta)
  \label{eq:jupitermarg}
\end{equation}
where $\lambda = \delta/\gamma$, and the univariate functions $f$ and $g$ are
\begin{align}
  f(\lambda) &= y^T y - (A^Ty)^T (A^TA + \lambda L)^{-1} (A^T y)\label{eq:f}\\
  g(\lambda) &= \log \det\left( A^TA + \lambda L \right). \label{eq:g}
\end{align}
Both these functions are analytic, monotonic, and very mildly behaved over a wide range of arguments, as can be seen in Figure~\ref{fig:fg} (periodic case). 
\begin{figure}
\centerline{
\includegraphics[width=0.5\textwidth]{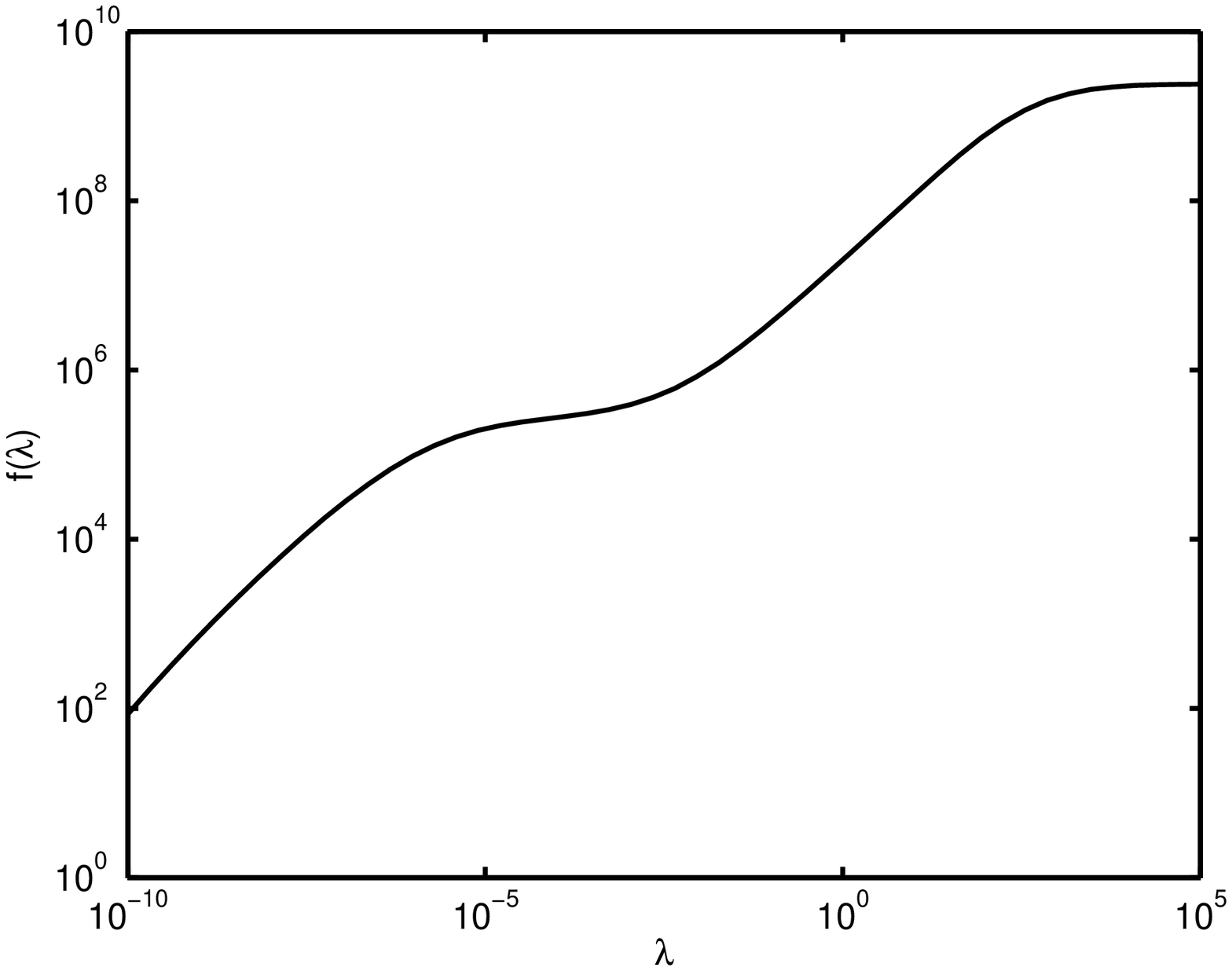} 
\includegraphics[width=0.5\textwidth]{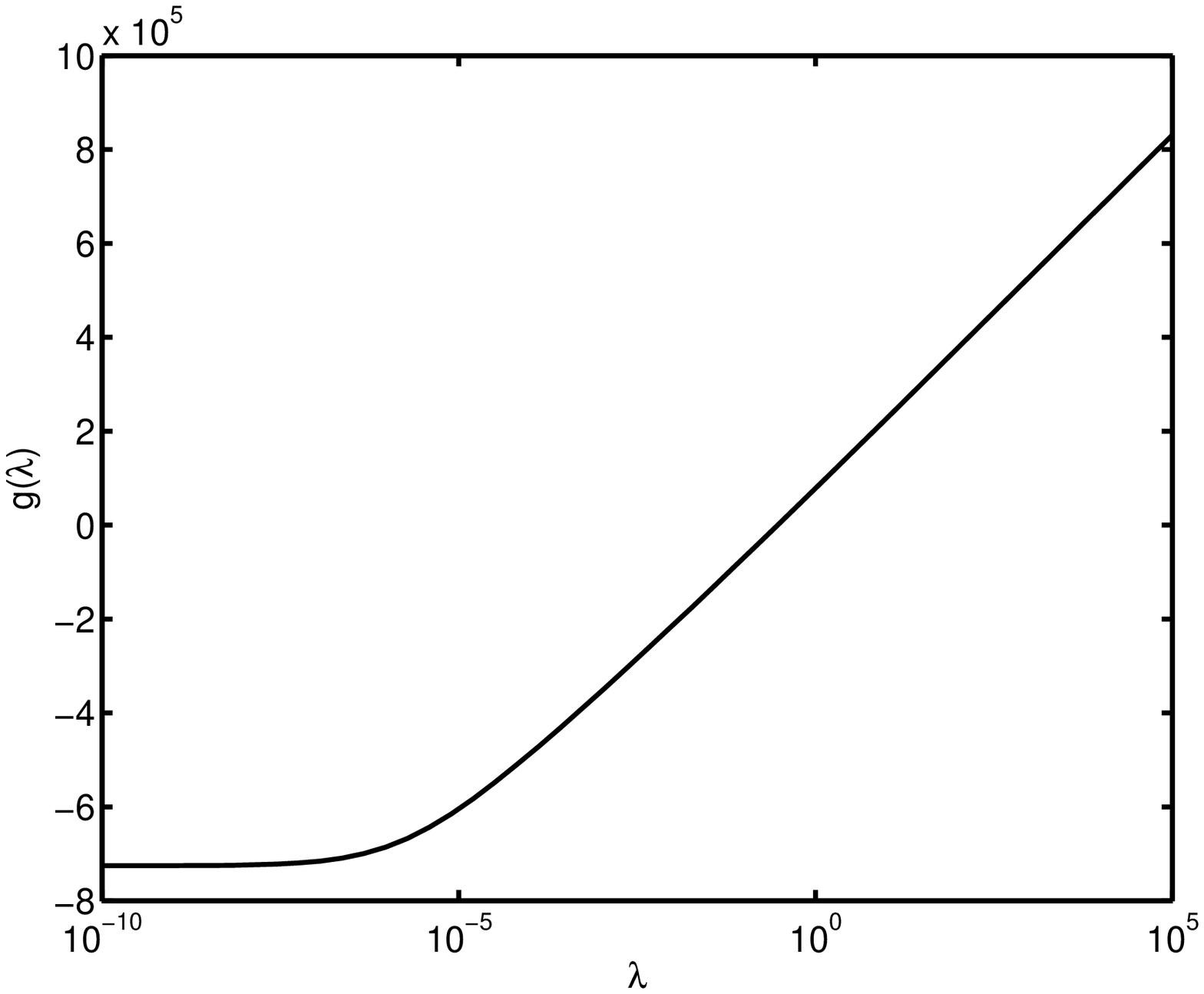} 
}
\caption{Functions $f(\lambda)$ and $g(\lambda)$ for the periodic $p\times p$ image model}
\label{fig:fg}
\end{figure}
One would therefore expect that efficient computation is possible.
We develop efficient calculation of the (difference of) functions $f$ and $g$, and sampling from $ \pi(\theta|y)$, in Section~\ref{sec:periodic} for periodic boundary conditions, and for the general case in Section~\ref{sec:nonperiodic}.

\subsubsection{Full conditional for $x$}
\label{sec:fullcond}
For the linear-Gaussian hierarchical model~\eqref{eqs:lgbm}, the full conditional for $x$ may be readily determined~\cite{higdon2006,SimpsonLindgrenRue2012} as the multivariate normal distribution
\begin{equation}
x|y,\theta\sim\Normal\left(\mu_{x|y,\theta},Q_{x|y,\theta}^{-1}\right)\label{eq:xgivenytheta}
\end{equation}
where
\begin{eqnarray*}
\mu_{x|y,\theta} & = & \mu+\left(Q+A^{T}\Sigma^{-1} A\right)^{-1}A^{T}\Sigma^{-1}\left(y-A\mu\right)\\
Q_{x|y,\theta} & = & Q+A^{T}\Sigma^{-1} A.
\end{eqnarray*}
Again, we have omitted the dependence of matrices on $\theta$ for brevity.

For the Jupiter example we have
\begin{equation}
\label{eq:xblock}
	x |\theta,y \sim \mbox{N}\left( (A^TA + (\delta/\gamma) L)^{-1} A^T y, (\gamma A^T A + \delta L)^{-1} \right).
\end{equation}
An independent sample from this distribution may be computed  by solving 
\begin{equation}
  \left( \gamma A\trans A + \delta L\right) x = \gamma A\trans y + w
  \label{eq:rto}
\end{equation}
where $w=v_1+v_2$ with independent $v_1\sim\Normal\left(0,\gamma A\trans A\right)$ and $v_2\sim\Normal\left(0,\delta L\right)$. 

In the periodic case, the FFT diagonalizes all matrices so the cost of sampling the random variables and solving~\eqref{eq:rto} is $\mathcal{O}(n)$ operations in the transform domain with a further $\mathcal{O}(n \log n)$ operations for each FFT. In the general case the cost of sampling random variables remains at $\mathcal{O}(n)$ operations since the covariance for $v_1$ is factorized and $A$ is sparse, while the neighborhood definition of $L$ which is the covariance for $v_2$ allows it to be written as a sum over cliques of $2\times 2$ rank-$1$ matrices; we call the resulting $\mathcal{O}(n)$ sampler for $v_2$ \emph{assembly by cliques} because it uses the same decomposition as \emph{assembly by elements} of a finite-element method stiffness matrix. Hence, in both cases, the cost of sampling $x |\theta,y$ is dominated by the cost of solving~\eqref{eq:rto}, which is precisely the same linear solve required in the generalized deconvolution step~\eqref{eq:gendeconv}.

We have written~\eqref{eq:rto} in the form given by Bardsley~\cite{BardsleyRTO}, who called the method \emph{randomize then optimize}. The calculation in~\eqref{eq:rto} had been previously used, to our knowledge, by Oliver, He and Reynolds in 1996~\cite{oliver1996conditioning} under the moniker \emph{randomized maximum likelihood} and more recently in~\cite{OrieuxFG12} as \emph{perturbation-optimization}. This formulation for sampling from the full conditional for the latent field was also used by Wikle \emph{et al.}~\cite{WMNB2001} who solved the system using early termination of a conjugate gradient solver, noting that the quality of the \emph{approximate} sample could be controlled by the convergence criterion. Inexact solution of~\eqref{eq:rto} followed by a Metropolis-Hastings accept/reject step has also been considered; \cite{GMI} used early termination of a conjugate gradient solver, while~\cite{NortonFoxAP} established convergence properties when using a finite number of steps of a linear iterative solver. 

\section{Comparative description of sampling algorithms}
\label{sec:comparative}
We present two sampling algorithms in addition to the MTC algorithm. The first is a block Gibbs sampler~\cite{higdon2006,BardsleyRTO,OrieuxFG12} that has been presented as an efficient way to sample from high-dimensional linear-Gaussian inverse problems. The second is the `one-block' algorithm introduced by Rue and Held~\cite{RueHeld} that almost always improves on the block Gibbs sampler in statistical efficiency while having similar cost per iteration. Finally we present the MTC algorithm that has the same statistical efficiency as the one-block algorithm but with lower computational cost per iteration.

\subsection{Block Gibbs sampler\label{sec:Gibbs}}
Random-walk MCMC sampling of the posterior distribution suffers from slow mixing due to high correlations within the distribution over images, and between the image and hyperparameter. Slow mixing within sampling of the high-dimensional image may be alleviated by updating the unknown image in a single Gibbs step, i.e., treating the image components in a single block. 
This is feasible because the full conditional distribution over image $x$, given everything else, is the multivariate normal~\eqref{eq:xgivenytheta} or \eqref{eq:xblock} so independent samples may be drawn using efficient methods from numerical linear algebra for solving the system~\eqref{eq:rto}.  

Block Gibbs sampling proceeds by drawing from the full conditional distributions over image $x$ then hyperpaprameters $\gamma$ and $\delta$, repeatedly in sequence. Hence it is also necessary to have the full conditional distributions for $\gamma|x,\delta,y$ and $\delta|x,\gamma,y$ available in a form that can be sampled. This is possible when using \emph{conjugate} prior distributions over the hyperparameters. For the Jupiter example we follow~\cite{higdon2006,BardsleyRTO} and use the Gamma distributions  $\gamma \sim \Gamma(\alpha_\gamma,\beta_\gamma)$ and $\delta \sim \Gamma(\alpha_\delta,\beta_\delta)$, i.e.,
\begin{eqnarray}
\pi(\gamma) & \sim & \gamma^{\alpha_\gamma-1} \exp(-\beta_\gamma \gamma) \label{eq:lprior} \\
\pi(\delta) & \sim & \delta^{\alpha_\delta-1} \exp(-\beta_\delta \delta) \label{eq:dprior}
\end{eqnarray}
in which $\alpha_{\gamma},\beta_{\gamma},\alpha_{\delta},\beta_{\delta}$ are constants chosen to make the hyperprior distributions ``relatively
uninformative''~\cite{BardsleyRTO}.  We use $\alpha_\gamma = \alpha_\delta = 1$ and $\beta_\gamma = \beta_\delta = 10^{-4}$ as in \cite{BardsleyRTO}, although this hyperprior distribution can be viewed as informative for the scale of $\delta$, as mentioned above.
The resulting conditional distributions over hyperparameters are
\begin{eqnarray}
\gamma|x,\delta,y & \sim & \Gamma\left(\tsfrac{m}{2}+\alpha_{\gamma},\tsfrac{1}{2}\left\Vert Ax-y\right\Vert ^{2}+\beta_{\gamma}\right)\label{eq:condl}\\
\delta|x,\gamma,y & \sim & \Gamma\left(\tsfrac{n}{2}+\alpha_{\delta},\tsfrac{1}{2}\left\Vert Ax-y\right\Vert ^{2}+\beta_{\delta}\right)\label{eq:condd}.
\end{eqnarray} 

The block Gibbs sampler is the only sampler we present that requires a mathematically convenient form for the hyperprior distribution. We see this as a major disadvantage of (block) Gibbs sampling, for the reasons given at the end of Section~\ref{sec:Bayes}. Despite that objection, we will use the hyperprior distributions~\eqref{eq:lprior} and \eqref{eq:dprior} in all sampling algorithms to enable exact comparisons.

A block Gibbs sampler may then be implemented by cycling through sampling from the conditional distributions in~\eqref{eq:xblock},
\eqref{eq:condl} and \eqref{eq:condd} to get Algorithm \ref{alg:Gibbs}, as implemented in~\cite{BardsleyRTO,higdon2006}.
\begin{algorithm}
\label{alg:Gibbs}
at state $x$, $\theta = (\gamma,\delta)$

draw $x|\gamma,\delta,y \sim \mbox{N}\left( (A^TA + (\delta/\gamma) L)^{-1} A^T y, (\gamma A^T A + \delta L)^{-1} \right)$

draw $\gamma|x,\delta,y  \sim  \Gamma(\frac{m}{2}+\lambda_{\gamma},\frac{1}{2}\left\Vert Ax-y\right\Vert ^{2}+\beta_{\gamma})$

draw $\delta|x,\gamma,y  \sim  \Gamma(\frac{n}{2}+\lambda_{\delta},\frac{1}{2}\left\Vert Ax-y\right\Vert ^{2}+\beta_{\delta})$

\caption{Gibbs sampling algorithm with blocking of the latent field}
\end{algorithm}
Most of the computational cost per iteration is contained in the draw from the large Gaussian latent field, since that requires a solve of~\eqref{eq:rto}.

In practical inverse problems, it is found that the block Gibbs sampler requires about $10^2$ to $10^3$ iterations per effectively independent sample~\cite{BardsleyRTO,BoneH,GMI}. This is about $2$ orders of magnitude improvement over na\"ive random-walk MCMC directly on the posterior distribution.

It is well known that the statistical efficiency of (block) Gibbs sampling is dependent on parameterization and that the rate of convergence may be improved with an appropriate re-parameterization, see e.g.~\cite{PRS2007}.  However, computational efficiency in the Jupiter example is not necessarily improved since re-parametrization will, in general, require three linear solves per iteration rather than one, increasing the cost per iteration. Recent results also show that re-parametrization can lead to dimension independent mixing~\cite{NotBardsleyStuart}. This does not imply dimension independent computational cost since the block Gibbs sampler remains a geometrically convergent algorithm that requires at least one linear solve per iteration, whose computational cost increases with dimension. In contrast, we will find that the MTC sampler requires just one linear solve per \emph{independent} sample, beyond a fixed setup phase. 

\subsection{One-block sampler\label{sec:oneblock}}

The one-block algorithm~\cite[\S 4.1.2]{RueHeld} is usually feasible whenever the calculations required for the block Gibbs sampler are feasible. Further, the one-block sampler almost always has better statistical efficiency than Gibbs, including after re-parametrization, and does not require a special form for the hyperprior distribution, so should be preferred over the block Gibbs sampler in most circumstances. 

The one-block algorithm is so named because the hyperparameter $\theta$
and latent field $x$ are blocked together within a single Metropolis-Hastings
accept-reject step. In this scheme a candidate hyperparameter $\theta'$
is drawn according to some proposal distribution $q\left(\theta'|\theta\right)$, typically a random-walk, then~\eqref{eq:xgivenytheta}
is utilized to draw $x'$ conditioned on $\theta'$ and $y$. The composite proposal $\left(x',\theta'\right)$ is then accepted with probability (w.p.)
given by the usual Metropolis-Hastings rule on the joint posterior distribution. The
effective proposal distribution for the composite proposal is then
\[
q\left(x',\theta'|x,\theta\right)=\pi\left(x'|\theta',y\right)q\left(\theta'|\theta\right),
\]
and the one-block algorithm may be written as Algorithm \ref{alg:oneblock}. 

\begin{algorithm}
\label{alg:oneblock}
at state $x,\theta$

draw $\theta'\sim q\left(\theta'|\theta\right)$

draw $x'\sim\pi\left(x'|\theta',y\right)$

accept $\left(x',\theta'\right)$ w.p. $\alpha\left(\left(x,\theta\right)\rightarrow\left(x',\theta'\right)\right)=1\wedge\dfrac{\pi\left(x',\theta'|y\right)\pi\left(x|\theta,y\right)q\left(\theta|\theta'\right)}{\pi\left(x,\theta|y\right)\pi\left(x'|\theta',y\right)q\left(\theta'|\theta\right)}$

otherwise reject

\caption{One-block algorithm}
\end{algorithm}

\begin{lemma}
\label{lem:oneblock}
In Algorithm~\ref{alg:oneblock} the transition kernel for the hyperparameter $\theta$ is in detailed balance with the marginal posterior distribution for $\theta|y$.
\end{lemma}
\begin{proof}
\[
\pi\left(x,\theta|y\right)=\pi\left(x|\theta,y\right)\pi\left(\theta|y\right)
\]
so the Metropolis-Hastings ratio is (assuming no densities are zero)
\begin{eqnarray*}
\dfrac{\pi\left(x',\theta'|y\right)\pi\left(x|\theta,y\right)q\left(\theta|\theta'\right)}{\pi\left(x,\theta|y\right)\pi\left(x'|\theta',y\right)q\left(\theta'|\theta\right)} & = & \dfrac{\pi\left(x'|\theta',y\right)\pi\left(\theta'|y\right)\pi\left(x|\theta,y\right)q\left(\theta|\theta'\right)}{\pi\left(x|\theta,y\right)\pi\left(\theta|y\right)\pi\left(x'|\theta',y\right)q\left(\theta'|\theta\right)}\\
 & = & \dfrac{\pi\left(\theta'|y\right)q\left(\theta|\theta'\right)}{\pi\left(\theta|y\right)q\left(\theta'|\theta\right)}.
\end{eqnarray*}
\end{proof}

Thus, the chain in $\theta$ targets $\pi\left(\theta|y\right)$, as if we have been able to integrate out the
latent field. 
Because this chain makes steps in the \emph{marginal} posterior distribution for $\theta$, rather than the conditional for $\theta$ given the current $x$, it takes larger steps and the chain in $x,\theta$ converges more rapidly to the joint posterior distribution (see~\cite{PRS2007} noting that in ill-posed inverse problems the data $y$ is necessarily `weakly informative' in many dimensions of $x$). In particular, high correlation between hyperparameters and the image, that partially motivates re-parametrization of the Gibbs sampler, is irrelevant to mixing of the one-block algorithm. Figure 4.1 of~\cite{RueHeld} illustrates this. 

Evaluation of the acceptance probability in Algorithm~\ref{alg:oneblock} requires evaluating a ratio of determinants. The ratio of posterior distributions typically does not present difficulties since, as in the Jupiter example, the scaling of determinants with respect to hyperparameters is simple. However, the ratio of determinants associated with the full conditional for $x$ in the proposal presents the traditional difficulty. In work by Rue and colleagues (e.g.~\cite{RueHeld,SimpsonLindgrenRue2012}) it is assumed that sampling from a multivariate normal is performed using Cholesky factorization and hence the required determinants are available at a further cost of $n$ multiplications. However, in very large problems computing the Cholesky factorization is prohibitively expensive and in general iterative solvers of~\eqref{eq:rto} are most efficient. Then the required determinants are not directly available. In Section~\ref{sec:nonperiodic} we present an efficient method for calculating the required ratio of determinants when using iterative solvers for large problems. 

\subsection{MTC sampler\label{sec:MTCdetails}}

We implement the MTC sampler in Algorithm~\ref{alg:margcond} by performing Metropolis-Hastings MCMC sampling directly from $\pi\left(\theta|y\right)$, shown in Algorithm~\ref{alg:MH}, and only after obtaining an (effectively) independent sample $\theta \sim \pi(\theta|y)$ do we then draw $x \sim \pi(x|\theta,y)$ to get an (effectively) independent sample $(x,\theta) \sim \pi(x,\theta|y)$ from the full posterior distribution. 
\begin{algorithm}
\label{alg:MH}
at state $\theta$

draw $\theta'\sim q\left(\theta'|\theta\right)$

accept $\theta'$ w.p. $\alpha(\theta \rightarrow \theta' )=1\wedge\dfrac{\pi\left(\theta'|y\right)q\left(\theta|\theta'\right)}{\pi\left(\theta|y\right)q\left(\theta'|\theta\right)}$

otherwise reject

\caption{Metropolis-Hastings algorithm on $\pi(\theta|y)$}
\end{algorithm}
If we use the same proposal distribution $q\left(\theta'|\theta\right)$ as the one-block Algorithm~\ref{alg:oneblock}, then by Lemma~\ref{lem:oneblock} both MTC and one-block will generate the same chain over $\theta$ and hence these algorithms have the same statistical efficiency. However, MTC evaluates the Metropolis-Hastings ratio directly using $\pi\left(\theta|y\right)$, and thus avoids the cost of the linear solve required to draw from the full conditional for $x$ in each MCMC iteration. 
 
Reduced computational cost is possible when it is possible to cheaply evaluate the ratio 
\( {\pi\left(\theta'|y\right)}/{\pi\left(\theta|y\right)} \)
required in the Metropolis Hastings acceptance probability. For the general linear-Gaussian  model~\eqref{eqs:lgbm} this involves evaluating  ratios of determinants of $\Sigma^{-1}$, $Q$ and $Q+A^T\Sigma^{-1}A$ in~\eqref{eq:overlooked}, and differences of arguments of the exponential (which are also required in the one-block algorithm). Efficient calculation of these terms is developed in Sections~\ref{sec:periodic} (periodic boundary conditions)  and \ref{sec:nonperiodic} (general case). 

\section{Numerical comparisons for the periodic model\label{sec:periodic}}
In this Section we develop specific computational schemes and present numerical results for all algorithms applied to the Jupiter deblurring problem, in the simplified setting when periodic boundary conditions are assumed for the unknown image of size $p\times p$, as used for regularized inversion computed in Section~\ref{sec:reg}. This allows us to use the computational cost of regularized Fourier deconvolution as a benchmark, being a standard  efficient method for image deblurring~\cite{NumRec}. The FFT also transforms determinants into products to give, at worst, an $\mathcal{O}(n)$ calculation.

\subsection{MCMC sampling from $\pi(\theta|y)$}
We present two algorithms for the MCMC over $\theta|y$ in the MTC sampler: the first (Option 1) is a `no think' implementation using a random walk MCMC and the $\mathcal{O}(n)$ evaluation of determinants; the second (Option 2) implements Metropolis-within-Gibbs over a re-parametrization of $\theta$ that is more efficient and uses the efficient calculation of functions $f$ and $g$ detailed in the Appendix that is potentially $\mathcal{O}(1)$ in image size, provided that the number of terms $\lambda Z_i$ in the interval $[c,c^{-1}]$ is $\mathcal{O}(1)$ (see the Appendix). Then the on-line cost of sampling from $\pi(\theta|y)$ will be $\mathcal{O}(1)$ as $n \rightarrow \infty$, so the only on-line computation that depends on image size will be the single solve of~\eqref{eq:rto} that generates an independent image sample.

{\bf Option 1}. 
From current state $\theta = (\gamma,\delta)$ propose $\theta' = (\gamma',\delta')$ according to
\begin{eqnarray*}
	\gamma' | \gamma & \sim & \Normal(\gamma,w_\gamma^2) \\
	\delta' | \delta & \sim & \Normal(\delta,w_\delta^2),
\end{eqnarray*}
that defines the proposal density 
$$
	q(\theta'|\theta) \propto \exp \left( -\frac{1}{2w_\gamma^2}(\gamma'-\gamma)^2 - \frac{1}{2w_\delta^2} (\delta'-\delta)^2 \right),
$$
and proceed as in Algorithm~\ref{alg:MH}. This MCMC simplifies to the Metropolis algorithm as the proposal density function is symmetric, hence the Hastings ratio $q\left(\theta|\theta'\right)/q\left(\theta'|\theta\right)$ always equals $1$. The simple $\mathcal{O}(n)$ calculation of determinants is used. 

A useful guide is to tune  $w_\gamma$ and $w_\delta$ until the acceptance ratio is approximately $0.5$ in low dimensions and $0.25$ in high dimensions \cite{RobertsLectureNotes}.  We used $w_\gamma = 2.34\times 10^{-3}$ and $w_\delta = 17.28 \times 10^{-7}$, which correspond to approximately $1.8$ times the standard deviation of $\gamma | y$ and $\delta |y$, respectively.

{\bf Option 2.}  This option uses a more efficient technique for evaluating $f$ and $g$ and a Metropolis-within-Gibbs algorithm with bespoke Gibbs directions to obtain a near-optimal implementation of MTC.  Instead of sampling from $\delta | \gamma,y$ and $\gamma | \delta,y$ we use the polar coordinates $\phi = \tan^{-1}(\delta / \gamma)$ and $r = \sqrt{\delta^2 + \gamma^2}$ and sample from $r | \phi,y$ and $\phi | r,y$.  We do this because we are able to directly draw independent samples from $r | \phi,y$ since 
$$
	r | \phi,y \sim \Gamma\left( \frac{n}{2} + \alpha_\gamma + \alpha_\delta, \frac{\cos \phi}{2}f(\tan \phi) + \beta_\gamma \cos \phi + \beta_\delta \sin \phi \right).
$$
We then use one iteration of a Metropolis algorithm to sample from $\phi | r,y$ that has density function
\begin{flalign*}
  &\pi(\phi|r,y) \propto \pi(\delta,\gamma|y) \propto & \\
  & (\cos\phi)^{\alpha_\gamma-1} (\sin\phi)^{n/2 + \alpha_\delta-1} \exp \left( -\frac{1}{2} g(\tan\phi) - \frac{r \cos \phi}{2} f(\tan\phi) - \beta_\gamma r \cos\phi - \beta_\delta r \sin \phi \right),&
\end{flalign*}
using the symmetric proposal
$$
	q(\phi'|\phi) = N(\phi';\phi,w_2^2)
$$
with $w_2 = 10^{-5}$ chosen so that the acceptance rate is approximately $0.44$, being the optimum acceptance rate in one-dimension when the target is Gaussian \cite{RobertsLectureNotes}. This defines the stochastic iteration for sampling $\theta|y$ in Algorithm~\ref{alg:MwG}.  

\begin{algorithm}
\label{alg:MwG}
at state $\theta$

draw $r | \phi,y \sim \Gamma\left( \frac{n}{2} + \alpha_\gamma + \alpha_\delta, \frac{\cos \phi}{2}f(\tan \phi) + \beta_\gamma \cos \phi + \beta_\delta \sin \phi \right)$

draw $\phi' \sim q(\phi'|\phi)$

accept $\phi'$ w.p. $\alpha(\phi \rightarrow \phi' )=1\wedge\dfrac{\pi\left(\phi'|r,y\right)q\left(\phi|\phi'\right)}{\pi\left(\phi|r,y\right)q\left(\phi'|\phi\right)}$

otherwise reject

\caption{Metropolis-within-Gibbs algorithm on $\pi(\theta|y)$ for the Jupiter example with directions $\phi = \tan^{-1} (\delta/\gamma)$ and $r = \sqrt{\delta^2 + \gamma^2}$}
\end{algorithm}
Evaluation of the ratio of marginal densities uses the series expansion of functions $f$ and $g$ detailed in the Appendix.

\subsection{Numerical results\label{sec:sim}}
To compare computational efficiencies we implemented the four options for sample-based inference, described above, applied to the same posterior distribution, that is, using the conjugate distributions over the hyperparameters given by~\eqref{eq:condl} and \eqref{eq:condd}. We initialized each Markov chain at $\gamma = 0.218$ and $\delta = 5.15 \times 10^{-5}$ which is the mode of the marginal distribution function $\pi(\theta|y)$ over $\theta$, found using MATLAB's {\tt fminsearch} function. For the one block algorithm we also used initial $x = (\gamma A^T A + \delta L)^{-1} \gamma A^T y$ and the same proposal as MTC Option 1.  

In order to evaluate accurate statistics we computed chains of length $10000$, then, by inspection, discarded a burn in of length $60$ for block Gibbs and $20$ for the other algorithms.  All computation was performed in MATLAB R2012b using a Lenovo X230 laptop with an Intel CORE i5 processor.
We used the MATLAB's {\tt fft2} function to diagonalize the action of $A$ and $L$, and the algorithms were implemented in the transform domain to reduce the number of FFTs required.

Figure \ref{fig:samp-hist-jupiter} shows the image component of a single sample from the posterior distribution, and marginal histograms for the hyperparameters $\gamma$, $\delta$ and the effective regularization parameter $\delta/\gamma$, using the MTC Option 2 algorithm. 
Note that all algorithms we use are provably convergent with the same distributional limit, so the particular algorithm we used is not actually important.

A single posterior sample provides \emph{unbiased} estimates of any property of the posterior image, in contrast to the regularized inverse that necessarily produces biased estimates. In this sense, a single posterior image sample could be viewed as superior to the regularized solution. The posterior mean image is more usually thought of as the Bayesian counterpart to the regularized inverse; we present a mean image in Section~\ref{sec:nonperiodic}.

\begin{figure}
\includegraphics[width=0.45\textwidth]{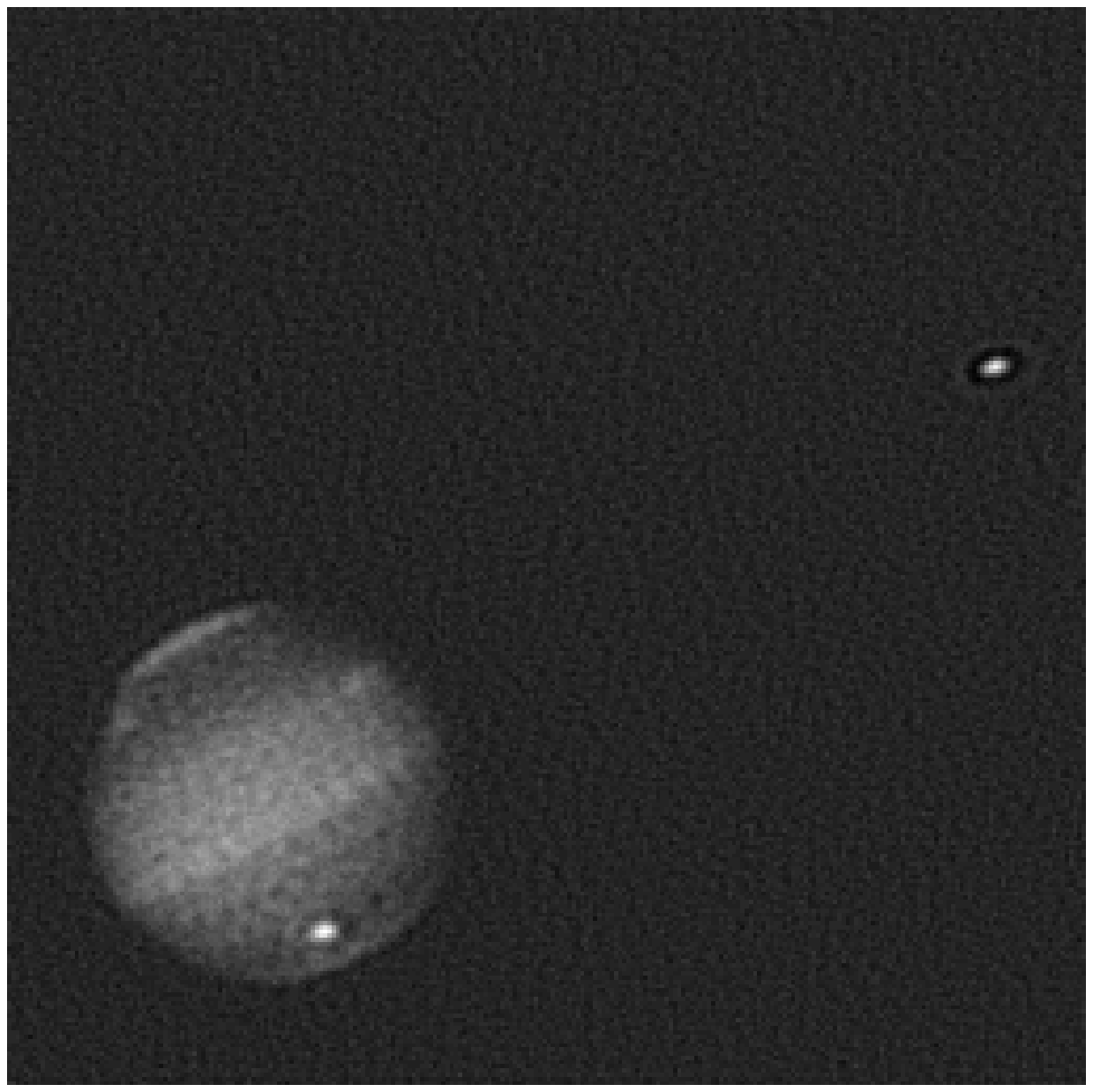}
\includegraphics[width=0.54\textwidth]{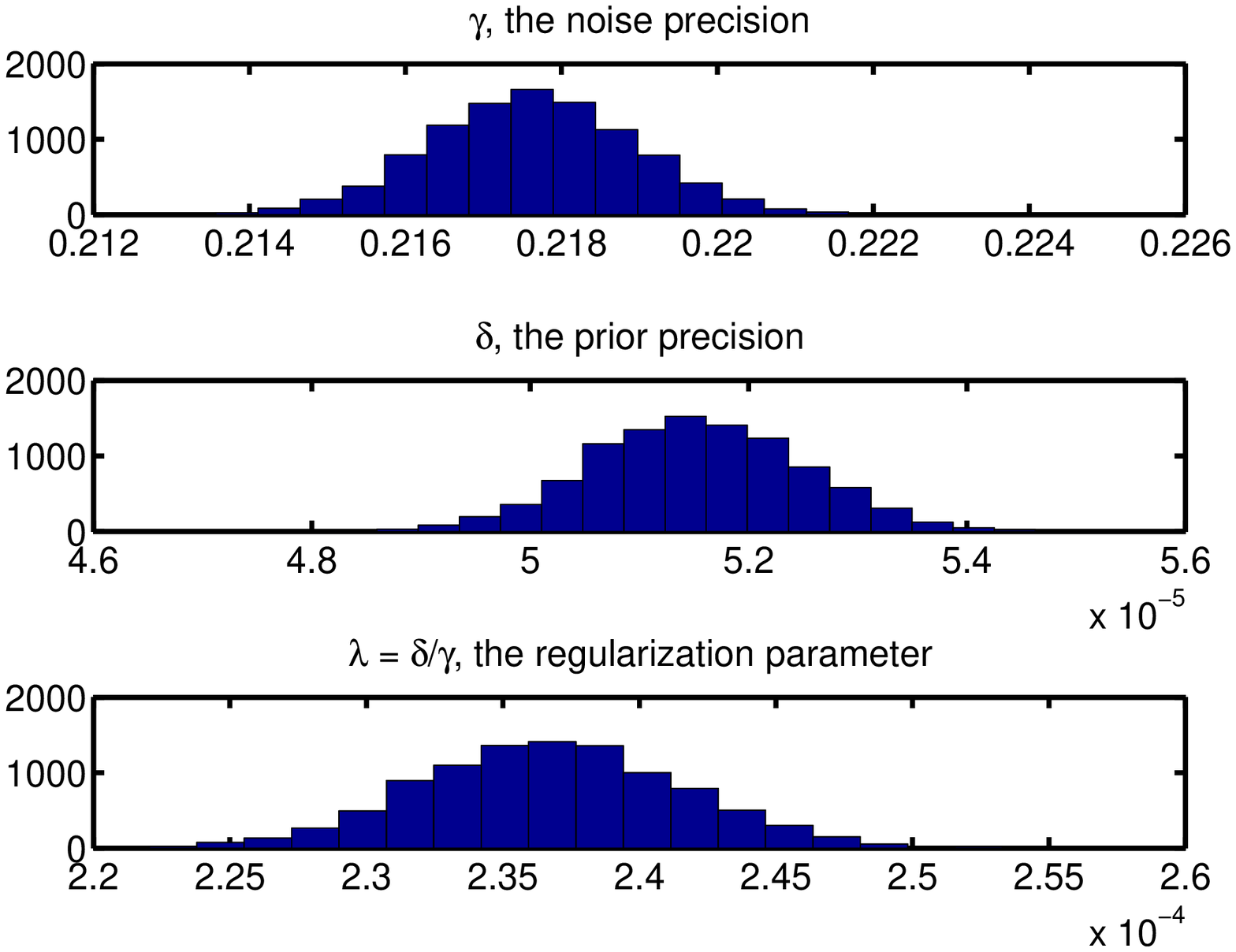}
\caption{Left; image component of one posterior sample, calculated using the MTC Option 2.  Right; marginal posterior histograms for $\gamma$, $\delta$ and $\delta/\gamma$.}
\label{fig:samp-hist-jupiter}
\end{figure}

We also calculated sample-based values of $\|Ax-y\|$ and $\sqrt{x^T L x}$ and plotted these as crosses in Figure~\ref{fig:deconv-jupiter} to indicate the posterior distribution of these statistics. Note that the crosses are tightly clustered and lie away from the L-curve;  no regularized estimate provides a good posterior summary of these statistics\footnote{This observation is not surprising when one notes that the regularized estimate in~\eqref{eq:varreg} may be written $\hat{x}_\lambda = \arg \min_{x} x^T L x$ subject to $\lVert Ax - y \rVert^2 =c $ for various $c>0$. Hence \emph{all} posterior samples lie \emph{above} the L-curve. In the presence of noise one therefore expects the minimizer to display stochastic bias and to be an outlier in the posterior distribution for these statistics. We can conclude that optimization does not provide valid summary statistics of the Bayesian posterior distribution.}. The posterior distribution over effective regularization parameter $\delta/\gamma$ has a sharp peak at $2.37 \times 10^{-4}$, which differs significantly from the value suggested by the L-curve method.

\subsection{Computational efficiency}
We measure computational efficiency of algorithms by evaluating the \emph{computing cost per effective sample} (CCES), defined as~\cite{GMI}
$$
\mathrm{CCES} = \frac{ \tau T}{N}
$$
where $\tau$ is the integrated autocorrelation time (IACT) for the statistic of interest, $T$ is the total (on-line) compute time, and $N$ is the length of chain ($10^4$ in our case).  We use the definition
\begin{equation}
   \tau = 1 + 2 \sum_{k=1}^\infty \rho_k, \label{eq:IACT}
\end{equation}
where $\rho_k$ is the autocorrelation at lag $k$, which gives the length of the chain that has the same variance reducing power as one independent sample.  Thus, CCES is the compute time required to reduce variance in estimates by the same amount as one independent sample; smaller is better. 
We estimate $\tau$ using \emph{twice}\footnote{Two definitions of IACT are used in the literature, one being twice the other.  Physics literature~\cite{Sokal96} tends to define IACT as $\frac{1}{2}\tau$, whereas statistics literature~\cite{Geyer} uses the definition in~\eqref{eq:IACT}.} the value computed by Wolff's {\tt UWerr.m} code~\cite{Wolff} for each of $\gamma$, $\delta$, and $\lambda=\delta/\gamma$. Autocorrelation functions for $\lambda$ are shown in Figure~\ref{fig:autocorrelation}. 
\begin{figure}
  \begin{center}
    \includegraphics[scale=0.5]{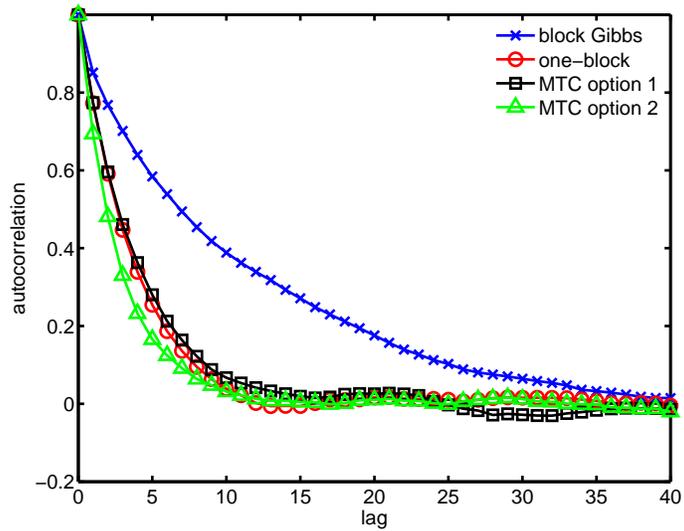}
  \end{center}
  \caption{ Autocorrelation of $\lambda = \delta/\gamma$ for all four sampling algorithms.}
\label{fig:autocorrelation}

\end{figure}

\begin{table}
\caption{Compute times and CCES (in seconds) and acceptance rates.}
\label{tab:1}
\begin{center}
\begin{tabular}{|p{3cm}||c|c|c|c|}
\hline
 & burn in & total time & acceptance rate & CCES for $\lambda$ \\
\hline \hline
Block Gibbs & 60 & 83.1 & 1 & 0.17 \\
One block & 20 & 127.8 & 0.33 & 0.090 \\
MTC Option 1 & 20 & 63.1 & 0.33 & 0.050 \\
MTC Option 2 & 20 & 11.8 & 0.46 & 0.015 \\
\hline
\end{tabular}
\end{center}
\end{table}

\begin{table}
\caption{Integrated autocorrelation times (in iterations) for three statistics of interest.}
\label{tab:2}
\begin{center}
\begin{tabular}{|p{3cm}||c|c|c|}
\hline
 & $\gamma$ & $\delta$ & $\lambda=\delta / \gamma$ \\
\hline \hline
Block Gibbs & 1.6 & 22.3 & 21.0 \\
One block & 7.8 & 6.7 & 7.1 \\
MTC Option 1 & 7.6 & 7.8 & 7.9 \\
MTC Option 2 & 2.1 & 5.0 & 5.7 \\
\hline
\end{tabular}
\end{center}
\end{table}

Total compute times, CCES, and IACT are shown for each sampling algorithm in Tables~\ref{tab:1} and \ref{tab:2}. We see that statistical efficiency (IACT) and computational efficiency (CCES) follow the pattern expected from theoretical considerations, with MTC Option 2 being clearly the most efficient in both measures.  Option 2 improves on Option 1 in IACT because of the improved MCMC for $\theta|y$ while compute time per iteration is also reduced due to the efficient evaluation of $f(\lambda)$ and $g(\lambda)$.


In Table \ref{tab:2},  block Gibbs has IACT for $\delta$  much larger than for $\lambda$, whereas the other algorithms have similar IACT for $\delta$ and $\lambda$.  This suggests that high correlations between $x$, $\gamma$, and $\delta$ reduce the efficiency of block Gibbs. Results for the block Gibbs sampler are shown without re-parametrization that could potentially improve its statistical efficiency. Looking at the IACT for $\lambda$ in Table~\ref{tab:2} one would expect that better parametrization could reduce the IACT of block Gibbs in this example to that of MTC Option 2, since that algorithm is not affected by correlations between $x$ and $\theta$, and the re-parametrization of $\theta$ gives efficient sampling. That is, statistical efficiency  of block Gibbs could be improved by a little more than factor of 3. Since the re-parametrization of Gibbs increases the cost per iteration by a factor of about 3, re-parametrization would only slightly reduce the CCES for block Gibbs. This better parametrization of the hyperparameters could also be applied to the one-block algorithm to reduce its IACT by a small amount. Thus, we expect that after a re-parametrization of \emph{both} algorithms, the one-block algorithm will remain roughly twice as computationally efficient as block Gibbs sampling, as shown in Table~\ref{tab:1}. This agrees with the suggestion in~\cite{RueHeld} that the one-block algorithm is to be preferred over Gibbs in most circumstances.

In Table \ref{tab:2} we also see that IACT for both the one-block algorithm and MTC Option 1 are approximately the same which agrees with theory that the algorithms are statistically equivalent. MTC Option 1 has a smaller CCES than does one-block because evaluating the acceptance ratio for MTC Option 1 has approximately half the computing cost of the one-block algorithm. 



\begin{table}
\caption{Compute time and number of solves required for a regularized image or an independent image sample.}
\label{tab:3}
\begin{center}
\begin{tabular}{|p{2.3cm}||c|c||c|c||c|}
\hline
 & time to $\lambda$ & solves to $\lambda$ & time to $x$ & solves to $x$ & total time \\
\hline \hline
Regularization & 0.52 & 200 & 0.0024 & 1 & 0.52 \\
Block Gibbs & 0.85 & 102 & 0.0020 & 0 & 0.85 \\
One block & 0.44 & 34 & 0.0020 & 0 & 0.44 \\
MTC Option 1 & 0.23 & 36 & 0.0096 & 1 & 0.24 \\
MTC Option 2 & 0.037 & 0 & 0.0082 & 1 & 0.045 \\
\hline
\end{tabular}
\end{center}
\end{table}

Finally, in Table \ref{tab:3} we compare sampling to regularization. {\emph{Time to $\lambda$} is the compute time for finding $\lambda$ by constructing an L-curve for regularization, and is 
$$
	\frac{T}{N} ( (\mathrm{burn}) + 2\tau)
$$
for sampling algorithms. \emph{Time to $x$} is any additional compute time required to produce a deconvolved image.  For regularization this is the time to solve~\eqref{eq:gendeconv}, for the one-block algorithm it is the cost of a single inverse FFT, and for the MTC algorithms it is the cost of drawing a sample from $x|y,\theta$ by solving~\eqref{eq:rto}.  

MTC Option 2 is a whole order of magnitude faster than regularization at computing a candidate for the deconvolved image of Jupiter.  Regularization spends the majority of computing time constructing an L-curve to find a suitable regularization parameter, requiring $200$ solves, and then solves the linear system~\eqref{eq:gendeconv} to estimate the deconvolved Jupiter by the MAP estimate of $x|y,\theta$.  MTC Option 2 is an order of magnitude faster at computing  $\lambda = \delta/\gamma$ as an independent draw from $\pi(\theta|y)$, and then draws an independent sample from $x|y,\theta$ by a single solve of~\eqref{eq:rto}.   

We also show the number of \emph{solves to $\lambda$}, which is $((\mathrm{burn}) + 2\tau)$ for the block Gibbs, one-block and MTC Option 1 algorithms, and the subsequent \emph{solves to $x$} for each algorithm. Since each algorithm is dominated by the cost of the linear solves, these figures give relative compute costs in the general case when linear solves are more expensive than $\mathcal{O}(n)$.

\section{MTC sampling for non-periodic models\label{sec:nonperiodic}}
We now present an implementation of MTC sampling for deblurring Jupiter using an image model that includes a 16 pixel wide band of nuisance pixels around the image region and without the simplifying assumption of periodic boundary conditions. We will see that, after an initial off-line computing phase, samples from $\pi(\theta|y)$ can be computed cheaply and independent of the image size. 


Let $x \in \mathbb{R}^n$ where $n = (256+32)^2$ be the true image augmented with a border of $16$ pixels, and impose zero Dirichl\'{e}t boundary conditions beyond that.  We compute the action of $A$, $A^T$, and $L$ directly by convolution using MATLAB's {\tt conv2} function. Since $A \in \mathbb{R}^{m \times n}$ and $L \in \mathbb{R}^{n \times n}$ are sparse matrices, operation by these matrices requires $\mathcal{O}(n)$ operations but we never assemble or factorize the matrices.

The marginal posterior distribution $\pi(\theta|y)$ is given by \eqref{eq:jupitermarg} and we use conjugate prior distributions in \eqref{eq:lprior} and \eqref{eq:dprior}, as before.  
To sample from $\pi(\theta|y)$ we implemented a Metropolis-within-Gibbs algorithm with Gibbs directions $\gamma$ and $\lambda = \delta/\gamma$.  We can directly draw independent samples of $\gamma | \lambda,y$ since
$$
	\gamma | \lambda,y \sim \Gamma \left( \frac{m}{2} + \alpha_\delta + \alpha_\gamma, \frac{1}{2}f(\lambda) + \beta_\gamma + \beta_\delta \lambda \right), 
$$
and we use the random walk Metropolis algorithm with proposal $\lambda' | \lambda \sim N(\lambda,w_3^2)$ with $w_3 = 10^{-4}$ to sample from $\lambda | \gamma,y$ which has conditional density function
$$
	\pi(\lambda | \gamma ,y) \propto \lambda^{n/2 + \alpha_\delta -1} \exp \left( -\frac{1}{2} g(\lambda) - \frac{\gamma}{2} f(\lambda) - \beta_\delta \gamma \lambda \right).
$$
Efficient implementation of this algorithm then depends on the ability to efficiently evaluate $f(\lambda)$ (for the Gibbs step), $f(\lambda')-f(\lambda)$, and $g(\lambda')-g(\lambda)$.  We found that quartic Taylor series expansions of $f$ and $g$, beyond the zeroth term, about the mode $\lambda_0=\arg\max_{\lambda}\pi(\lambda,\gamma|y)$, gave sufficiently accurate results for the present example. 

Writing $B = A^T A + \lambda_0 L$, the derivatives of $f$ are 
\begin{equation}
  f^{(r)}(\lambda_0) = (-1)^{r+1} k! (A^Ty)^T (B^{-1}L)^r B^{-1} (A^Ty), \qquad r=1,2,\dotsc.
  \label{eq:fderiv}
\end{equation}
Using the identity~\cite[p.29]{gohberg2000}
\[\log( \det ( I + t F)) = \sum_{r=1}^\infty \frac{(-1)^{r+1}}{r!} \operatorname{tr}(F^r) t^r\]
the derivatives of $g$ are
\begin{equation}
  g^{(r)}(\lambda_0) = (-1)^{r+1} \operatorname{tr}( (B^{-1}L)^r ), \qquad r=1,2,\dotsc.
  \label{eq:gderiv}
\end{equation}

We evaluated Monte Carlo estimates of each trace in~\eqref{eq:gderiv} by exploiting the identity $\operatorname{tr}( (B^{-1}L)^r ) = \E[ z^T (B^{-1}L)^r z]$ where each $z_i\stackrel{\text{iid}}{\sim}\operatorname{Unif}\left(\{-1,1\}\right)$, see e.g. \cite[\S 6]{Meurant2009}. This calculation makes \emph{even} order Taylor expansion most convenient as the compute work required to evaluate an odd derivative allows the next even derivative to be evaluated for free. The accuracy of Monte Carlo estimates and the number of terms in the Taylor series may be determined so that Monte Carlo and truncation errors are smaller than the relative error inherent in performing the linear solve in finite precision, though this gives a conservative bound. 

We used MATLAB's {\tt gmres} solver function for each linear solve, restarted every $25$ iterations, with relative residual tolerance of $10^{-3}$.  A tighter tolerance resulted in very long solve times when $\lambda$ is small.  The condition number of $B$ is approximately $10^4$. 
We found that just $4$ samples of $z$ gave sufficiently small relative Monte Carlo error. 
Hence, to obtain the quartic Taylor expansions of $f$ required $3$ linear solves, and $g(\lambda)-g(\lambda_0)$ required  $4 \times 4 = 16$ linear solves.  

We used an iterative procedure to find the mode of $\pi(\gamma,\lambda|y)$, by computing quartic Taylor expansions of $f$ and $g(\lambda)-g(\lambda^{(r)})$ about $\lambda^{(r)}$, then using these expansions in place of $f$ and $g$ to obtain $\lambda^{(r+1)}=\arg\max_{\lambda}\pi(\lambda|\gamma,y)/\pi(\lambda^{(r)}|\gamma,y)$.  We terminated this algorithm when $|\lambda^{(r+1)}-\lambda^{(r)}|/|\lambda^{(r)}| < 10^{-2}$.  
Starting from $\lambda^{(0)} = 5 \times 10^{-3}$ this required $3$ iterations and $57$ linear solves to converge to $\lambda_0 = 4.4879 \times 10^{-3}$.  We tested convergence by also starting from $\lambda^{(0)} < \lambda_0$ and found the algorithm converged to the same $\lambda_0$.

After computing quartic Taylor expansions of $f$ and $g$, the cost of computing a Markov chain of length $10^4$ was negligible compared to a linear solve.  Note that block Gibbs and the one-block algorithm (not simulated here) both require a linear solve per sample, so would require $10^4$ solves to compute a Markov chain of the same length. The $74 = 57 + 3 + 16$ linear solves required before computing the Markov chain for MTC is comparable with, though smaller than, the 200 solves required to determine $\lambda$ for regularization.

Figure~\ref{fig:np-samp-hist-jupiter} shows histograms of $\gamma$, $\delta$ and $\lambda=\delta/\gamma$ values computed from the MTC sampler.  We also computed the marginal posterior mean image, shown in Figure \ref{fig:np-samp-hist-jupiter}, as follows. The posterior expectation of any function $h(x)$ may be written
\[
\text{E}_{x,\theta |y}\left[ h\left(x \right) \right] = \text{E}_{\theta |y}\left[ \text{E}_{x |\theta,y}\left[ h\left(x \right) \right] \right]
\] 
which is a weighted sum in $\theta$ of expectations over full conditionals in $x$.  This allows efficient calculation when the inner expectation $\text{E}_{x |\theta,y}\left[ h\left(x \right) \right]$ is cheap to evaluate; the outer expectation is low dimensional and may be easily \emph{estimated} via a Monte Carlo integral utilizing samples from the MCMC over $\theta |y$, or \emph{evaluated} via numerical integration once the marginal posterior distribution over $\theta$ is well determined. We implemented the latter. In the linear Gaussian problem the full conditional for $x$ is Gaussian so any moment may be evaluated this way, i.e. for polynomial $h$. 

The mean in the Jupiter example further simplifies to
\begin{equation}
\label{eq:mean}
	\E[x|y] = \int (A^T A + \lambda L)^{-1} A^T y \, \pi(\lambda|y) \, \dd \lambda
\end{equation}
with weights for the numerical integration given by the marginal posterior histogram for $\lambda$.  This requires as many linear solves as there are bins required for the histogram for $\lambda$. Note that the integration in~\eqref{eq:mean} requires the solve in~\eqref{eq:gendeconv} rather than the solve in~\eqref{eq:rto} that also requires drawing random numbers. 

\begin{figure}
\includegraphics[width=0.43\textwidth]{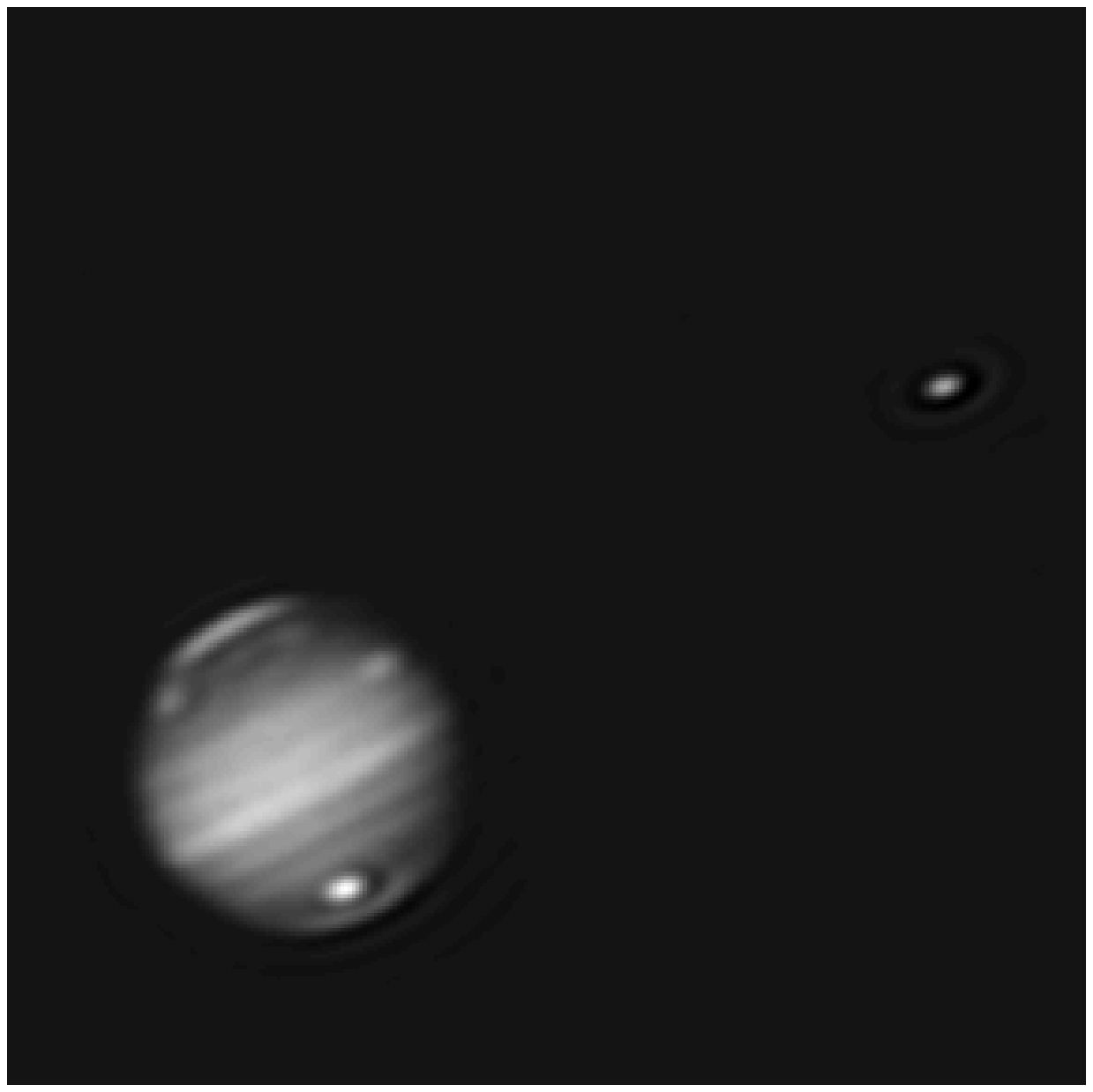}
\includegraphics[width=0.56\textwidth]{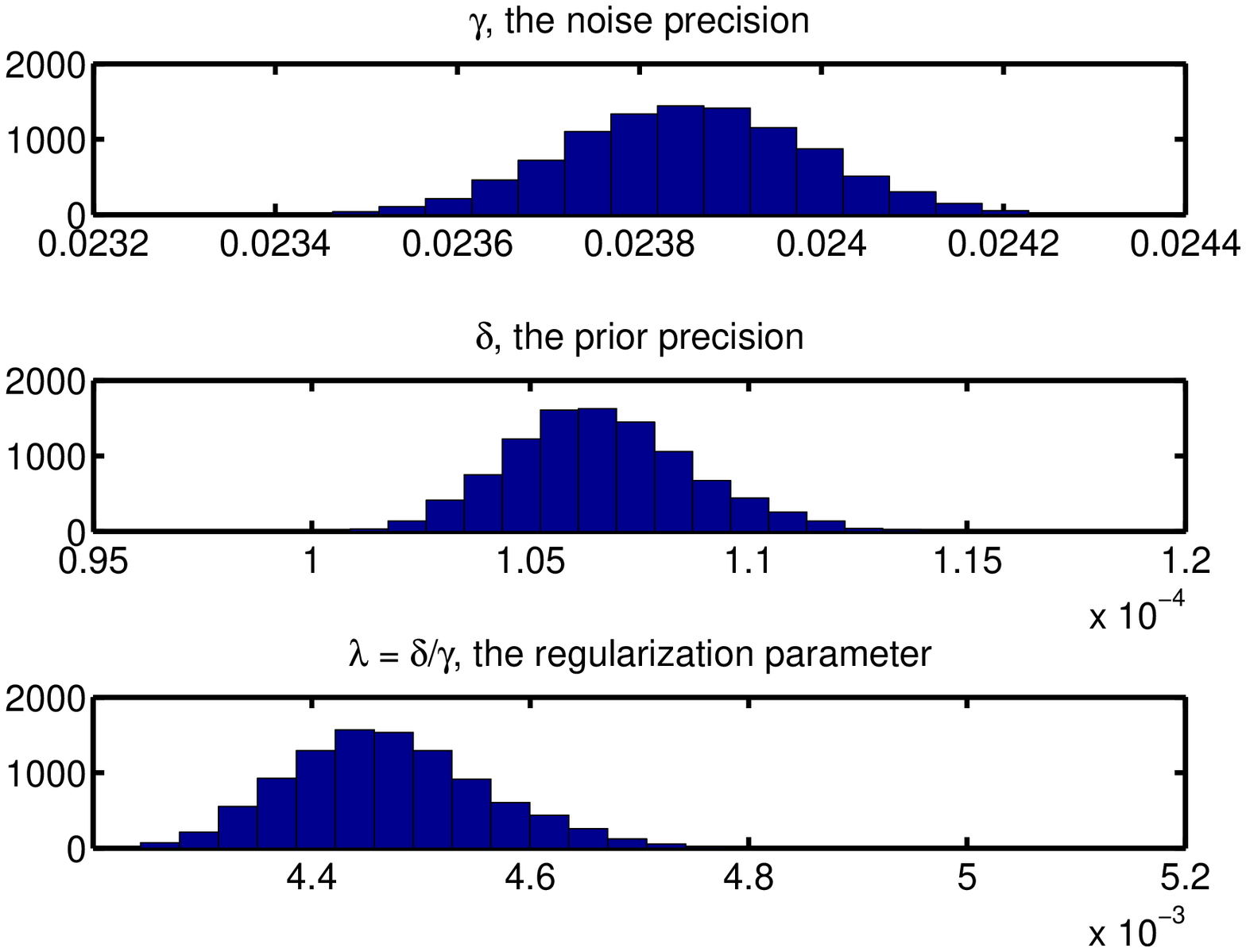}
\caption{Left; mean of posterior image with non-periodic boundary conditions.  Right; marginal posterior histograms for $\gamma$, $\delta$ and $\delta/\gamma$.}
\label{fig:np-samp-hist-jupiter}
\end{figure}

It is interesting to note that the posterior mean image shown in Figure \ref{fig:np-samp-hist-jupiter} gives a better deblurred image than the regularized solution in Figure~\ref{fig:deconv-jupiter}, with the bands of Jupiter, and other details, being more clearly defined. This improvement may be the result of better modeling of the image boundary employed in this example. Artificial ringing around the satellite remains evident and, as for regularization, reduction of this can be achieved by better modeling of the point-spread function. The marginal posterior variance of images may be evaluated to provide valid uncertainties, as indicated above, though we have not made that calculation for this example.

The posterior histogram for $\lambda$ in Figure~\ref{fig:np-samp-hist-jupiter} shows that $\lambda$ is effectively supported on $|\lambda-\lambda_0| < 3 \times 10^{-4}$. Since $f^{(5)}(\lambda_0) = 5.04 \times 10^{16}$ and $g^{(5)}(\lambda_0) = 4.2 \times 10^{16}$, the truncation error for $f$ and $g$ is approximately $3 \times 10^{-3}$.  
We also evaluated the Monte Carlo error in estimates of the trace for each derivative of $g$.  
This was $9.7 \times 10^{3}$ for $g'(\lambda_0)$ and $2.4 \times 10^6$ for $g''(\lambda_0)$.  Hence, the error in the linear term for $g(\lambda)-g(\lambda_0)$ is bounded by approximately 
$3$ and the error in the quadratic term is approximately 
$10^{-1}$.  The errors in the cubic and quartic terms are $10^{-3}$ and $10^{-5}$, respectively. Small relative errors in the linear terms correspond to a small relative scaling in the hyperparameters, which does not affect inference in $x$ when hyperpriors are uninformative in this scale near $\lambda_0$. Thus the dominant error is the quadratic term of $10^{-1}$, which corresponds to a small relative error in the \emph{variance} of the marginal posterior distributions over $\delta$ and $\gamma$. Using multiple centers in a piecewise expansion provides a simple mechanism for reducing the effect of errors to any desired level. 



\section{Discussion\label{sec:discussion}}
We considered posterior sampling for a canonical linear inverse problem with $n=65536$ and $n=82944$ unknowns using a Bayesian hierarchical model with Gaussian likelihood and GMRF prior. In the computed  example we used a conjugate hyperprior distribution  to enable block Gibbs sampling. By using the observation in~\cite{BardsleyRTO} that sampling from the full conditional distribution over the latent field may be performed by a regularized solve (plus samples from standard normals), we focused effort on sampling from the marginal posterior distribution over hyperparameters. Our main contribution is to show that an MCMC over that marginal distribution can be cheap and fast, by reducing the computation within iterations to the evaluation of two smooth functions of a single variable, to give an algorithm that we call marginal then conditional (MTC) sampling.

In the first computed example we compared the computational cost of the MTC sampler with the cost of regularized Fourier inversion and also a range of MCMC sampling algorithms that have been applied to the linear-Gaussian problem; the MTC sampler outperforms all these algorithms in the sense that the compute time required to generate an independent sample from the posterior is less than required for the alternative MCMC algorithms, and also an order of magnitude less than the cost of regularized inversion. The latter result is, perhaps, the most surprising as MCMC is often viewed as necessarily slower than deterministic algorithms. However, that is clearly not the case, as we have shown. We also computed a second example to demonstrate how  MTC sampling may be implemented in the more general setting, where $A^TA$ and $L$ are not simultaneously diagonalizable using the FFT. By counting linear solves that dominate computational cost, we saw that the cost of calculating the posterior mean image was still less than the cost of the regularized inverse, by about a factor of 2 in the second computed example.

Posterior statistics that we evaluated differed significantly from the values given by the regularized solution. Figure~\ref{fig:deconv-jupiter} shows that the posterior distribution for $\sqrt{x^TLx}$ and data misfit $\|Ax-y\|$ is supported well away form the L-curve, so \emph{no} value of regularizing parameter gives a meaningful summary of these statistics. The same conclusion holds for $\lambda=\delta/\gamma$.  This demonstrates that analyses of inverse problems that purport to be `Bayesian' but then calculate solutions by optimization, and hence actually implement regularization, are evaluating summary statistics that may not even be in the support of the posterior distribution. 

The  decomposition used in the MTC sampler, in Algorithm~\ref{alg:margcond}, is widely known in statistics as the \emph{marginal algorithm}. An example is the `linchpin variable sampler' presented in~\cite{AcostaHuberJones}; that algorithm was defined with a draw from the full conditional for $x$ per iteration of an MCMC on the marginal posterior over hyperameters. The MTC sampler differs in that the MCMC over hyperparameters is run until an effectively independent sample is generated, and only then is a sample generated from the full conditional for $x$. Agapiou \emph{et al.}~\cite{NotBardsleyStuart} also implemented the marginal algorithm, sampling from the marginal posterior over the hyperparameter $\delta$, with $\gamma$ assumed known. They viewed the marginal algorithm as the ``gold standard'' finding, as we do here, that it is optimal for statistical efficiency, though  also referred to it as ``prohibitively expensive for large scale inverse problems''. We have been able to improve on that situation by implementing efficient computation within the MTC algorithm that enables inference over \emph{all} hyperparameters and that outperforms other sampling algorithms in both statistical and computational efficiency, especially for large image size. Indeed, we think that the computational scheme in MTC may be close to optimal as the computational cost per independent sample is dominated, for large $n$, by the cost of a single linear solve. 

Agapiou \emph{et al.}~\cite{NotBardsleyStuart} restricted their consideration to Gaussian prior models with trace class covariance in the infinite-dimensional limit, and finite-dimensional discretizations. 
We followed~\cite{Vogel,BardsleyRTO} by using a scaled graph Laplacian precision in two dimensions. This does not generate a trace class covariance, which shows up as the log singularity in fundamental solutions of the Laplacian in two dimensions. Thus the effective prior covariance function that we use has a logarithmic singularity, which seems undesirable from a modeling perspective. The singularity also leads to mesh dependence in solutions; this problem is not so evident when using a structured mesh, as in our examples, since truncation of the singular covariance is roughly the same across the mesh and so the variance is roughly stationary across the mesh. However, when using unstructured meshes the variance typically varies dramatically across the mesh, and so is unsuitable for applications. For these reasons, we also advocate restricting models to trace class covariances.

Efficient computation in MTC relied on being able to evaluate the univariate functions $f$ and $g$ in~\eqref{eq:f} and \eqref{eq:g}. When the functions $f$ and $g$ are well defined in the limit $n\rightarrow\infty$, pre-evaluation of these functions allows sampling independent $\theta|y$ in the $\infty$-dimensional case, with a single linear solve required to generate an independent posterior sample $x,\theta|y$, using, for example, an iterative solver directly on function space~\cite{nevanlinna}. We have presented computational schemes that are sufficient to perform efficient calculation in the examples provided. However, we expect that these computational schemes can be significantly improved upon, particularly for large $n$ and in other applications, and see this as a potentially fruitful topic for future research in computational UQ.

\section*{Acknowledgements}
The authors are grateful to J. Andr\'{e}s Christen for insightful comments, and to H\"avard Rue for his 2008 case study of Tokyo rainfall data that motivated this work. This work was supported by Marsden contract UOO1015.

\appendix
\section{Expansions of $f(\lambda)$ and $g(\lambda)$ in terms of $\lambda$}

For selected tolerance $\epsilon > 0$ and $s \in \mathbb{N}$ choose $c_f,c_g \in (0,1)$ such that 
$c_f^{s+1} \nrm{y}^2 = \epsilon$ and $c_g^{s+1} n = \epsilon$.  In practice we should tune $c_f$, $c_g$, and $s$ to gain efficiency.

Let $\hat{y} := \operatorname{DFT}(y)$, and let $\hat{A}$ and $\hat{L}$ be the vectors associated with using FFTs to evaluate $A$ and $L$ times a vector (they are the diagonal of the diagonalised matrices).  Define $Z_i := \hat{L}_i / |\hat{A}_i|^2$ for all $i$.  Then using Parseval's theorem we obtain
$$
	f(\lambda) = \frac{1}{n} \sum_{i=1}^n |\hat{y}_i|^2 \frac{\lambda Z_i}{1+\lambda Z_i}.
$$
Since the determinant of a matrix is the product of its eigenvalues we also obtain
$$
	g(\lambda) = \sum_{i=1}^n \log|\hat{A}_i|^2 + \sum_{i=1}^n \log(1+\lambda Z_i).
$$
Reorder the indices so that $\{Z_i\}_{i=1}^n$ is increasing and in $\mathcal{O}(n)$ operations, using cumulative sums, precompute a $s \times n$ matrix $S$, a $(s+1) \times n$ matrix $T$, and $s \times n$ matrices $U$ and $V$ with entries
$$
	S_{rq} = \sum_{j=1}^q \frac{|\hat{y}_j|^2}{n} Z_j^{r}, \quad
	T_{rq} = \sum_{j=q}^n \frac{|\hat{y}_j|^2}{n} Z_j^{-r}, \quad
	U_{rq} = \sum_{j=1}^q Z_j^{r}, \quad
	V_{rq} = \sum_{j=q}^n Z_j^{-r}.
$$
Also precompute scalar $a$ and $n$-vector $b$ with entries
$$
	a = \sum_{i=1}^n \log |\hat{A}_i|^2, \qquad
	b_s = \sum_{j=s}^n \log Z_i \qquad \mbox{for $s=1,\dotsc,n$.}
$$
The on-line calculation of $f(\lambda)$ and $g(\lambda)$ are described by the following lemmas.

\begin{lemma}
For each evaluation of $f(\lambda)$ define $m_1$ and $m_2$ such that
$$
	\lambda Z_i < c_f \mbox{ for $i \leq m_1$} \qquad \mbox{and} \qquad
	\lambda Z_i > c_f^{-1} \mbox{ for $i \geq m_2$.}
$$
Then for some $E$ satisfying $|E| \leq \epsilon$ we have
$$
	f(\lambda) = \sum_{i=m_1+1}^{m_2-1} \frac{|\hat{y}_i|^2}{n} \frac{\lambda Z_i}{1+\lambda Z_i}
	 + \sum_{r=1}^s (-1)^{r+1} \lambda^r S_{r m_1}
	 + \sum_{r=0}^s (-1)^{r} \lambda^{-r} T_{r m_2}
	 + E.
$$
\end{lemma}

\begin{proof}
We first split $f(\lambda)$ into three terms, $f(\lambda) = \mathcal{T}_1 + \mathcal{T}_2 + \mathcal{T}_3$.
The partial geometric series $\frac{1}{1+z} = (1-z+z^2 - \dotsc + (-1)^s z^{s-1}) + \frac{(-1)^{s+1} z^{s}}{1+z}$ for $|z|<1$ implies
$$
	\mathcal{T}_1 
	= \sum_{i=1}^{m_1} \frac{|\hat{y}_i|^2}{n} \frac{\lambda Z_i}{1+\lambda Z_i} 
	= \sum_{r=1}^s (-1)^{r+1} \lambda^r S_{r m_1} + E_1
$$
where $|E_1| \leq c_f^{s+1} \sum_{i=1}^{m_1} |\hat{y}_i|^2/n$ since $\lambda Z_i \leq c_f$ for $i\leq m_1$.

Similarly, for $\mathcal{T}_3$ we use the partial geometric series $\frac{1}{1 + z^{-1}} =  (1-z^{-1}+z^{-2}- \dotsc + (-1)^{s}z^{-s}) + \frac{(-1)^{s+1} z^{-(s+1)}}{1 + z^{-1}}$ for $|z|>1$ to obtain
$$
	\mathcal{T}_3
	= \sum_{i=m_2}^n \frac{|\hat{y}_i|^2}{n} \frac{1}{1+(\lambda Z_i)^{-1}} 
	= \sum_{r=0}^s (-1)^r \lambda^{r} T_{r m_2} + E_3
$$
where $|E_3| \leq c_f^{s+1} \sum_{i=m_2}^n |\hat{y}_i|^2/n$ since $(\lambda Z_i)^{-1} \leq c_f$ for $i \geq m_2$.  Hence result, noting that $E = E_1 + E_3$ satisfies $|E| \leq  c_f^{s+1} \sum_{i=1}^n |\hat{y}_i|^2/n = c_g^{s+1} \nrm{y}^2$.
\end{proof}

\begin{lemma}
For each evaluation of $g(\lambda)$, define $m_1$ and $m_2$ such that
$$
	\lambda Z_i < c_g \mbox{ for $i \leq m_1$} \qquad \mbox{and} \qquad
	\lambda Z_i > c_g^{-1} \mbox{ for $i \geq m_2$.}
$$
Then for some $E$ satisfying $|E| \leq \epsilon$ we have
\begin{multline*}
	g(\lambda) = a + b_{m_2} + (n-m_2+1)\log(\lambda) + \sum_{m_1+1}^{m_2-1} \log(1+ \lambda Z_i)\\ + \sum_{r=1}^s \frac{(-1)^{r-1}}{r} \lambda^r U_{r m_1}  + \sum_{r=1}^s \frac{(-1)^{r-1}}{r} \lambda^{-r} V_{r m_2}  + E.
\end{multline*}
\end{lemma}

\begin{proof}
We first split $g(\lambda)$ into four terms, $g(\lambda) = a + \mathcal{T}_1 + \mathcal{T}_2 + \mathcal{T}_3$.
	
Using the series expansion $\log(1+z) = \sum_{r=1}^s \frac{(-1)^{r-1}}{r} z^r + (-1)^{s} \int_0^1 \frac{t^s}{1+t} \dd t$ for any $|z|<1$ we obtain
$$
	\mathcal{T}_1 = \sum_{i=1}^{m_1} \log(1 + \lambda Z_i) 
	= \sum_{r=1}^s \frac{(-1)^{r-1}}{r} U_{r m_1} \lambda^r + E_1
$$
where
$
	|E_1| = \left| \sum_{i=1}^{m_1} (-1)^s \int_0^{\lambda Z_i} \frac{t^s}{1+t} \dd t  \right| \leq m_1 (\lambda Z_i)^{s+1} < m_1 c^{s+1}.
$

Similarly, for $\mathcal{T}_3$ we obtain
\begin{align*}
	\mathcal{T}_3 &= \sum_{i=m_2}^n \log(\lambda Z_i) + \sum_{i=m_2}^n \log\left(1 + \frac{1}{\lambda Z_i} \right) \\
	&= (n-m_2 + 1)\log \lambda + b_{m_2} + \sum_{r=1}^s  \frac{(-1)^r}{r} V_{rm_2} \lambda^{-r} + E_3
\end{align*}
where
$
	|E_3| = \left| \sum_{i=m_2}^n (-1)^s \int_0^{(\lambda Z_i)^{-1}} \frac{t^s}{1+t} \dd t \right| < (n-m_2+1) c^{s+1}.
$

Hence result, noting that $E = E_1 + E_3$, so $|E| < |E_1| + |E_3| < n c_g^{s+1}$.  
\end{proof}

\bibliographystyle{abbrv}
\bibliography{SLIP}

\end{document}